\documentclass[a4paper, onecolumn, 10pt, unpublished]{quantumarticle}
\pdfoutput=1

\usepackage[utf8]{inputenc}
\usepackage[english]{babel}
\usepackage[T1]{fontenc}

\usepackage{amsmath}
\usepackage{hyperref}
\usepackage{tikz}
\usepackage{lipsum}

\usepackage{siunitx}
\usepackage{layouts}
\usepackage{comment}
\usepackage[numbers,sort&compress]{natbib}
\usepackage{booktabs}
\usepackage{multirow}

\usepackage[labelformat=simple]{subcaption}

\usepackage{amsthm}
\usepackage{amsmath, amssymb}
\usepackage{optidef}
\usepackage{cleveref}

\crefname{figure}{figure}{figures}

\theoremstyle{definition}
\newtheorem{theorem}{Theorem}
\newtheorem{lemma}{Lemma}
\newtheorem{corollary}{Corollary}
\newtheorem{definition}{Definition}

\numberwithin{example}{section}
\numberwithin{definition}{section}
\numberwithin{lemma}{section}
\numberwithin{theorem}{section}
\numberwithin{equation}{section}
\newcommand{\bigO}{\mathcal{O}}

\usepackage{float} 
\usepackage{algorithm}
\usepackage{algorithmicx}
\usepackage{algpseudocode}
\algrenewcommand\algorithmicrequire{\textbf{Input:}}
\algrenewcommand\algorithmicensure{\textbf{Output:}}

\begin{document}

\title{Square-root Time Atom Reconfiguration Plan for Lattice-shaped Mobile Tweezers}

\author{Koki Aoyama}
\affiliation{Graduate School of Information Science and Technology, The University of Osaka, Suita, Osaka 565-0871, Japan}
\email{k-aoyama@ist.osaka-u.ac.jp}
\orcid{0000-0001-9137-3977}

\author{Takafumi Tomita}
\affiliation{Institute for Molecular Science, National Institutes of Natural Sciences, Okazaki 444-8585, Japan}
\affiliation{SOKENDAI (The Graduate University for Advanced Studies), Okazaki 444-8585, Japan}
\orcid{0000-0003-2126-1392}

\author{Fumihiko Ino}
\affiliation{Graduate School of Information Science and Technology, The University of Osaka, Suita, Osaka 565-0871, Japan}
\email{ino@ist.osaka-u.ac.jp}
\orcid{0000-0002-5757-7631}

\maketitle

\begin{abstract}
This paper proposes a scalable planning algorithm for creating defect-free atom arrays in neutral-atom systems.
The algorithm generates a $\bigO(\sqrt N)$ time plan for $N$ atoms by parallelizing atom transport using a two-dimensional lattice pattern generated by acousto-optic deflectors.
Our approach is based on a divide-and-conquer strategy that decomposes an arbitrary reconfiguration problem into at most three one-dimensional shuttling tasks, enabling each atom to be transported with a total transportation cost of $\bigO(\sqrt N)$.
Using the Gale--Ryser theorem, the proposed algorithm provides a highly reliable solution for arbitrary target geometries.
We further introduce a peephole optimization technique that improves reconfiguration efficiency for grid target geometries.
Numerical simulations on a 632$\times$632 atom array demonstrate that the proposed algorithm achieves a grid configuration plan that reduces the total transportation cost to 1/7 of state-of-the-art algorithms, while resulting in 32\%--35\% more atom captures.
We believe that our scalability improvement contributes to realizing large-scale quantum computers based on neutral atoms.
Our experimental code is available from \url{https://github.com/kotamanegi/sqrt-time-atom-reconfigure}.
\end{abstract}

\section{Introduction}

Advances in optical tweezer technology~\cite{cite:hardware_AOD_assembler2, cite:hardware_AOD_assembler3, cite:hardware_AOD_assembler}, which enable the trapping and manipulation of neutral atoms with programmable geometries and interactions, have established neutral-atom systems as a promising platform for fault-tolerant quantum computation~\cite{arXiv_2509.13247, PhysRevA.102.063107, cite:first_experiment_on_AOD, cite:hardware_AOD_assembler, cite:review_on_aod, cite:logical_quantum_processor, graham2022multi}.
However, the initial loading of atom arrays is a stochastic process~\cite{cite:first_experiment_on_AOD, cite:hardware_AOD_assembler, manetsch2024tweezerarray6100highly}, which inevitably leads to defects, or randomly vacant sites, with a probability of approximately 50\%--60\% per site.
These defects introduce computational errors that degrade the reliability of neutral-atom systems.
Therefore, a highly scalable method for creating large-scale defect-free atom arrays is a critical prerequisite for realizing scalable fault-tolerant quantum computing based on neutral atoms.

A primary approach to creating defect-free atom arrays is to reconfigure stochastically loaded atoms into a desired geometry.
This reconfiguration can be achieved using mobile optical tweezers~\cite{cite:review_on_aod, cite:optimal_routing, cite:logical_quantum_processor}, which can simultaneously capture and transport atoms through acousto-optic deflectors (AODs)~\cite{cite:review_on_aod, cite:optimal_routing, cite:logical_quantum_processor}.
However, AOD-controlled tweezers typically impose performance limitations on atom reconfiguration.
For instance, tweezer operations account for most of the reconfiguration time ($\sim \SI{200}{ms}$) in current AOD-based systems~\cite{cite:hardware_AOD_assembler, PhysRevA.102.063107}.
Reconfiguration time must be minimized to reduce atom losses because atoms have a limited trap lifetime owing to collisions with background gas~\cite{manetsch2024tweezerarray6100highly}.
Therefore, an efficient planning algorithm for atom reconfiguration is essential to generate a sequence of transportation operations that can rapidly assemble the desired defect-free atom array.

There are two main algorithmic strategies for reducing reconfiguration time.
The first is to reduce the total travel distance of all atoms, and the second is to exploit the inherent parallelism of transportation operations to move multiple atoms simultaneously.
The transportation time of each operation depends on the travel distance of the atom being moved.
Although several cost models, such as linear~\cite{cite:logical_quantum_processor, graham2022multi} and square-root (sqrt)~\cite{cite:multiple_tweezer_kagome, cite:single_tweezer_3}, have been proposed previously, the common objective is to minimize the total travel distance.
For an $n \times n$ square lattice of atom trapping sites, an upper bound on the total travel distance is given by $\bigO(n^3)$, where $n \in \mathbb{N}$ denotes the size of the static tweezer array.

With respect to exploiting transportation parallelism, previous planning algorithms~\cite{cite:multiple_tweezer_kagome, cite:multiple_tweezer_grid, cite:optimal_routing, cite:single_tweezer_1, cite:single_tweezer_2, cite:first_experiment_on_AOD, cite:single_tweezer_3} take advantage of the parallel shuttling capabilities of modern hardware.
For example, a two-dimensional (2D) tweezer array, or a set of mobile tweezers created by driving AODs in multitone mode~\cite{cite:crossed_AOD, cite:optimal_routing, cite:logical_quantum_processor}, can simultaneously transport multiple atoms on a 2D static tweezer array.
Effectively leveraging this parallel shuttling capability presents a considerable challenge because atoms that move simultaneously must form a regular pattern, whereas the initial atoms are stochastically generated with an irregular distribution.
Prior algorithms~\cite{cite:multiple_tweezer_kagome, cite:multiple_tweezer_grid, cite:first_experiment_on_AOD} address this challenge either by restricting transportation to a one-dimensional (1D) lattice pattern or by allowing complex transportation paths involving multiple relay points.
For example, the parallel sort-and-compression (PSC) algorithm~\cite{cite:multiple_tweezer_kagome} and the Tetris algorithm~\cite{cite:multiple_tweezer_grid}, when combined with a linear cost model, are capable of reconfiguring $N$ atoms in $\bigO (N)$ time, where $N=\bigO(n^2)$.
Specifically, these algorithms complete atom reconfiguration using $\bigO(\sqrt{N})$ transportation operations, with each operation moving an atom over a travel distance of $\bigO(\sqrt{N})$.
However, these approaches can reduce the degree of parallelism, make physical implementation more difficult, and in some cases fail to produce valid reconfiguration plans.
Therefore, a new approach is needed to efficiently exploit transportation parallelism while handling irregular initial atom patterns to fully utilize the $\bigO(N)$ transportation capacity available in current AOD-based systems.

In this paper, we propose a planning algorithm that can reconfigure $N$ atoms in $\bigO(\sqrt{N})$ time by exploiting the hardware-inherent parallelism of a 2D tweezer array.
The proposed algorithm achieves the target atom geometry by simultaneously transporting stochastically loaded atoms to their desired sites using a sequence of $\bigO(\sqrt{N})$ transportation operations.
We develop the planning algorithm based on a simple transportation model, in which each transportation operation captures atoms on a 2D lattice pattern and shifts them from their current sites to one of the four neighboring sites, namely up, down, left, or right.
The proposed algorithm offers the following technical contributions.
\begin{itemize}
    \item \textbf{High scalability:} To our knowledge, the proposed algorithm is the first to reconfigure $N$ atoms in $\bigO(\sqrt{N})$ time.
The generated plan achieves high scalability by fully using the AOD hardware capability to transport multiple atoms arranged in a 2D lattice pattern.
We explicitly demonstrate that the proposed algorithm can generate an efficient $\bigO(\sqrt{N})$ time plan for arbitrary target geometries.

    \item \textbf{High reliability:} The proposed algorithm consistently generates a valid reconfiguration plan that transports atoms from any arbitrary initial geometry to any arbitrary target geometry.
    Using the Gale--Ryser theorem~\cite{Gale1957-kk,Ryser1957-kk}, we prove that two row-wise shuttling operations and one column-wise shuttling operation are sufficient to solve arbitrary reconfiguration problems.

    \item \textbf{High feasibility:} The proposed algorithm uses a set of simple shift operations, which facilitates the implementation of 2D tweezer arrays under reasonable assumptions about AODs.
Specifically, the reconfiguration plan performs simple row-wise and column-wise shuttling operations to gather $N$ atoms by moving them toward the side of the atom array.
In addition, the proposed algorithm relies only on fundamental data structures, such as priority queues, and can therefore be easily integrated with real-time programmable hardware.
These advantages relax the hardware requirements for implementing the proposed algorithm on experimental platforms.
\end{itemize}

The proposed planning algorithm itself runs in $\bigO(N \log N)$ classical computation time (hereafter referred to as \textit{planning time}), which introduces additional overhead compared to previous algorithms~\cite{cite:multiple_tweezer_kagome, cite:multiple_tweezer_grid} with $\bigO(N)$ time complexity.
However, this planning overhead is predictable as a function of $N$ and is negligible for $N \leq 10^6$, where the overall reconfiguration time is dominated by the physical movement time of the optical tweezers.

From a feasibility perspective, we assume that the 2D tweezer system can simultaneously transport multiple atoms without incurring atom loss.
Current AOD-based systems also impose a limitation on the maximum number of mobile tweezers, owing to diffraction resolution constraints~\cite{cite:hardware_AOD_assembler2, cite:lukin_3000}.
Although these feasibility challenges must be addressed in the future, we believe that our theoretical results demonstrate sub-linear time scalability of the reconfiguration process, representing a crucial step toward scalable neutral-atom quantum computers.

\section{Preliminaries: Atom Reconfiguration Model} \label{sec:Preliminaries}

\begin{figure*}[t]
  \centering
  \begin{subfigure}[t]{0.48\textwidth}
    \centering
    \begin{tabular}{l}
      \subref{subfig:a} \\
      \includegraphics[keepaspectratio, width=\linewidth]{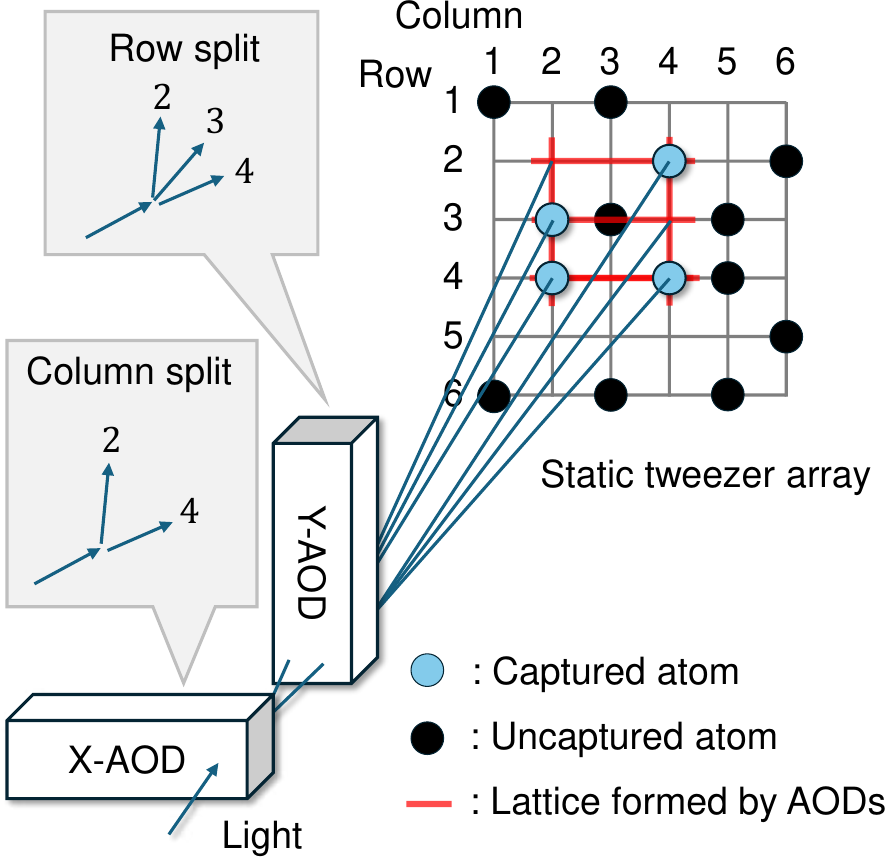}  
    \end{tabular}
    \phantomsubcaption\label{subfig:a}
 \end{subfigure}
 \hspace{0.08\textwidth}
 \begin{subfigure}[t]{0.42\textwidth}
    \centering
    \begin{tabular}{l}
       \subref{subfig:b} \\
       \includegraphics[keepaspectratio, width=\linewidth]{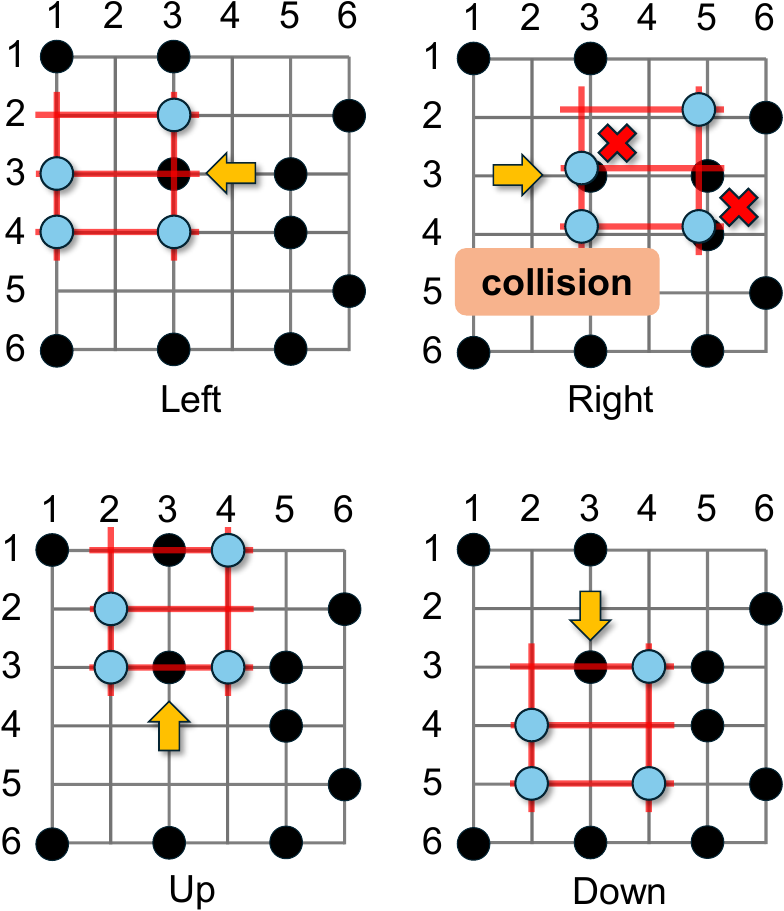}
    \end{tabular}
    \phantomsubcaption\label{subfig:b}
 \end{subfigure}
 \caption{Overview of the atom array system. The system comprises a static tweezer array and AOD tweezers, which capture atoms and transport them to other sites, respectively. \subref{subfig:a} A two-axis AOD is used to generate a two-dimensional lattice of light spots from the input beam.
 These spots are projected onto the atom array to enable atom capture and transport in the XY plane. The lattice can be dynamically reconfigured through a sequence of transport operations. \subref{subfig:b} Examples of transport operations applied to the atom array shown in (a). The atoms captured in rows $I=\{2, 3, 4\}$ and columns $J=\{2, 4\}$ are shifted from their current sites to adjacent sites in a specified direction (e.g., left, up, right, or down). In this case, the right-shift operation is prohibited to prevent collisions between moving and stationary atoms.} \label{fig:lattice_tweezer}
\end{figure*}

This study investigates a static-mobile tweezer system consisting of a static tweezer array and AOD tweezers, as shown in \Cref{fig:lattice_tweezer}.
The roles of these tweezers are as follows:
\begin{itemize}
    \item \textbf{The static tweezer array} generates an $n \times n$ square lattice of atom trapping sites, where each site can hold at most one atom and initially contains an atom with a probability $\alpha$ of approximately $50\%$. Without loss of generality, we assume a square lattice to facilitate better understanding but the proposed algorithm is applicable to rectangular lattices.
    \item \textbf{The AOD tweezers} are used to capture and transport atoms trapped in the static tweezer array, where a single beam of light is split by the AODs to form a 2D lattice beam. This assumption differs from previous PSC/Tetris algorithms~\cite{cite:multiple_tweezer_kagome,cite:multiple_tweezer_grid}, which rely on a 1D lattice beam. The 2D lattice beam captures and transports atoms located at the projected sites of the static tweezer array, and the captured atoms can be moved to other sites by adjusting the tone frequencies of the AODs.
\end{itemize}
The goal of atom reconfiguration is to realize a designated atomic geometry within the static tweezer array using the AOD tweezers.

\subsection{Tweezer System Model}
\label{subsec:model}

We abstract the static-mobile tweezer system as a 4-tuple $(n, S, \Phi, F)$, where $S$ denotes the set of valid atom geometries in the static tweezer array, $\Phi$ denotes the set of 2D tweezer operations, and $F$ denotes the objective function that evaluates a reconfiguration plan in terms of reconfiguration time.
Let $X=\{ 1, 2, \dots, n \}$ represent the set of row and column indices that specify the atom trapping sites on the lattice pattern.

\paragraph*{Configuration Space $S$.}
The configuration space $S$ consists of all atom geometries:
\begin{equation}
    S = \{ A \in \{0, 1\} ^{n\times n}\},
\end{equation}
where $A$ is a binary matrix representing an atom geometry, and the matrix element $A_{i,j} \in \{0,1\}$ indicates atom occupancy at the site with row $i$ and column $j$, where $i,j \in X$.
$A_{i,j} = 1$ means that the site $(i,j)$ is occupied by an atom, and $A_{i,j} = 0$ means that the site $(i,j)$ is vacant.
For example, a geometry $A$ with $n = 2$ has atoms located at sites $(1,1), (2,1)$, and $(2,2)$:
\begin{equation}
    A = \begin{pmatrix}
        1 & 0 \\
        1 & 1
    \end{pmatrix}.
\end{equation}

\paragraph*{2D Tweezer Operations $\Phi$.}

A 2D tweezer operation $\phi \in \Phi: S \rightarrow S$ is represented by a 3-tuple $(I,J,\delta)$, where $I \subseteq X$ and $J \subseteq X$ denote the sets of rows and columns, respectively, that define the lattice pattern generated by the mobile tweezers, and $\delta$ specifies the direction of atom transport.
In this model, an atom located at site $(i,j) \in I \times J$ can be shifted by one lattice site in one of four directions: left, right, up, or down.
We consider four 2D tweezer operators that act on all atoms within the lattice $(I,J)$:
\begin{equation}
\begin{aligned}
\mathrm{Left}(I,J)&: (i,j) \mapsto (i,j-1), \\
\mathrm{Right}(I,J)&: (i,j) \mapsto (i,j+1), \\
\mathrm{Up}(I,J)&: (i,j) \mapsto (i-1,j), \\
\mathrm{Down}(I,J)&: (i,j) \mapsto (i+1,j).
\end{aligned}
\end{equation}
\Cref{subfig:b} shows this concept with $(I,J) = (\{2,3,4\}, \{2,4\})$, where the 2D tweezers transport the atoms highlighted in blue.
The only constraint is that each site may contain at most one atom after a tweezer operation.
In the example shown in \Cref{subfig:b}, the $\mathrm{Right}(I,J)$ operation is prohibited because the sites marked with red crosses would contain two atoms following its application.

\paragraph*{Reconfiguration Plans $P$.}
Given the set of tweezer operations $\Phi$, the set of reconfiguration plans $P$ is defined as $P = \bigcup_{m=0}^\infty \Phi^m$, which intuitively means that the plan $p \in P$ consists of an arbitrary number $m$ of 2D tweezer operations.
A reconfiguration plan $p = (\phi_1, \phi_2, \dots, \phi_{|p|}) \in P$, where $|p|$ denotes the number of operations in the plan, transports atoms to other sites through these $|p|$ operations.
The initial geometry $A \in S$ is transformed to the final geometry $p(A) \in S$ by sequentially applying the operations in plan $p$:
\begin{equation}
    p(A) = \phi_{|p|}(\phi_{|p|-1}(\dots (\phi_1(A))\dots)).
\end{equation}
The concatenation of two plans, $p_1 = (\phi^{(1)}_1, \phi^{(1)}_2, \dots, \phi^{(1)}_{|p_1|})$ and $p_2 = (\phi^{(2)}_1, \phi^{(2)}_2, \dots, \phi^{(2)}_{|p_2|})$, is defined as
\begin{equation}
    p_2 \circ p_1 = (\phi^{(1)}_1, \phi^{(1)}_2, \dots, \phi^{(1)}_{|p_1|}, \phi^{(2)}_1, \phi^{(2)}_2, \dots, \phi^{(2)}_{|p_2|}).
\end{equation}

\paragraph*{Objective Function $F$.}
We estimate the reconfiguration time for plan $p$ using the following objective function:
\begin{equation}\label{eq:object}
    F(p, t_1, t_2) = |p| \, t_1 + |p| \, t_2,
\end{equation}
where $t_1 \in \mathbb{R^{+}}$ and $t_2 \in \mathbb{R^{+}}$ denote the elapsed time for one cycle of capturing and releasing atoms and the elapsed time for shifting captured atoms to neighboring sites, respectively.
The second $|p|$ in \Cref{eq:object} represents \emph{the total transportation cost} in this model.
This simple cost model assumes that (1) the reconfiguration time is dominated by the tweezer operation cost, $t_1$ and $t_2$, and (2) the cost of a tweezer operation is independent of the number of atoms moved in parallel.
Therefore, \Cref{eq:object} is independent of $N$ because AOD-based systems are capable of transporting multiple atoms simultaneously using a 2D tweezer operation.

\paragraph{Differences from previous models.}
Our model is a simplified version of earlier models~\cite{cite:atom_movr,cite:multiple_tweezer_kagome,cite:multiple_tweezer_grid}, summarized in \Cref{appendix:difficult_operations}.
The main distinctions between our model and the previous ones are outlined below.
\begin{itemize}
    \item \textbf{Simplified operations:} In our model, the 2D tweezer operations always move atoms in a single direction by a uniform distance. In contrast, previous models allow more complex operations that transport multiple atoms in different directions and through multiple relay points for long-distance transportation.
    \item \textbf{Simplified objective function:} The objective function in our model depends only on the number $|p|$ of tweezer operations, whereas in previous models it depends on both $|p|$ and the travel distance. Therefore, the cost minimization problem is replaced by an operation-minimization problem, simplifying the search for a valid reconfiguration plan.
\end{itemize}

We intentionally use a simple set of 2D tweezer operations to isolate the main challenge of massively parallel, collision-free atom transport from the added complexity of more elaborate 2D tweezer operations.
Previous models allow highly flexible but complex operations, such as compressions and expansions of the lattice pattern, to optimize the reconfiguration process; however, these operations require longer execution times and introduce additional constraints on the reconfiguration problem.
To our knowledge, it remains unclear whether the problem can be efficiently solved under such constraints.
Furthermore, recent hardware studies~\cite{cite:error_correction_AOD} have reported that lattice compression and expansion operations can lead to heating issues.
Compared with complex operations, our simple operations not only reduce constraints but also leverage the parallel transport capabilities of the 2D tweezer array.
We will later demonstrate that these simple operations are both powerful and sufficient for generating efficient reconfiguration plans.

According to the simplified operations, we also simplify the objective function used to evaluate reconfiguration plans.
Because the left, right, up, and down operations in our model involve the same travel distance and movement pattern of the 2D tweezers, we assume an identical cost for all operations, regardless of the movement direction or the specific rows or columns being shifted.
In contrast, the complex operations used in previous models can incur different costs because they move tweezers over varying travel distances.
As a result, previous models use a more accurate objective function that differentiates operational costs based on travel distance.
The details of the previous objective function are provided in \Cref{appendix:difficult_operations}.
Note that we use the previous model in our experimental evaluation to ensure a fair comparison with prior algorithms.

\subsection{Atom Reconfiguration Problem}
An instance of the reconfiguration problem is represented by a 4-tuple $(n, A^\textrm{(int)}, V, F)$, where $A^\mathrm{(int)} \in S$ denotes the initial geometry of the atom array, and $V: S \rightarrow \{0, 1\} $ is a verifier for the final atom geometry.
The initial geometry $A^\textrm{(int)}$ can be generated stochastically with a filling probability $\alpha \in [0, 1]$ as follows:
\begin{equation}
    A^\mathrm{(int)}_{i,j} = \left\{ \begin{array}{ll}
1, & \textrm{with probability } \alpha, \\
0, & \textrm{with probability } 1-\alpha.
\end{array}
\right.
\end{equation}
The objective of the atom reconfiguration problem is to find a plan $p$ satisfying
\begin{argmini!}
    {p \in P}          
    {F(p, t_1, t_2) \label{eq:prb-obj}} 
    {\label{eq:prb}}  
    {} 
    \addConstraint{V(p(A^\textrm{(int)}))}{=1. \label{eq:constr1}}
\end{argmini!}
That is, we aim to find a reconfiguration plan $p$ such that the verifier $V$ accepts $p$ as valid.
The verifier $V$ is defined according to the problem type.
We consider two types of reconfiguration problems, namely the arbitrary formation problem and the grid formation problem.
Note that in some reconfiguration problems, the target atom geometry $A^\textrm{(tgt)}$ is not explicitly given, as discussed later in this section.

\paragraph{Arbitrary formation problem.}
In the arbitrary formation problem, the objective is to construct a valid reconfiguration plan that realizes a specified target geometry $A^\mathrm{(tgt)}$, which must contain the same number of atoms as the initial geometry:
\begin{equation}
\sum^n_{i=1}\sum^n_{j=1} A^\mathrm{(tgt)}_{i,j} = \sum^n_{i=1}\sum^n_{j=1} A^\mathrm{(int)}_{i,j}.
\end{equation}
The verifier $V$ returns $1$ if and only if the final geometry matches the target configuration.
\begin{equation}
    V = \left\{ \begin{array}{ll}
1, & \textrm{if~} p(A^\textrm{(int)}) = A^\textrm{(tgt)}, \\
0, & \textrm{otherwise}.
\end{array}
\right.
\end{equation}

\paragraph{Grid formation problem.}
The objective of the grid formation problem is to find a valid reconfiguration plan that forms the largest possible $L \times L$ square atom array, where $L$ denotes the side length of the square.
The square size $L$ is given by
\begin{equation}\label{eq:define_L}
   L = \left\lfloor \sqrt{N} \right\rfloor,
\end{equation}
where $N = \sum_{i=1}^n \sum_{j=1}^n A^\mathrm{(int)}_{i,j}$.
The verifier $V$ returns $1$ if and only if the final geometry $p(A^\textrm{(int)})$ contains an $L \times L$ square atom array at any location.
\begin{equation}
    V = \left\{ \begin{array}{ll}
1, & \textrm{if~} \exists (i_0,j_0) \in X^2, \forall (i,j) \in X_1 \times X_2: p(A^\textrm{(int)})_{i,j} = 1, \\
0, & \textrm{otherwise},
\end{array}
\right. \label{cond:grid_reconf}
\end{equation}
where $i_0$ and $j_0$ denote the row and column of the top-left site in the atom array, respectively, $X_1=\{ i_0, i_0+1, \dots, i_0+L-1\}$, and $X_2= \{ j_0, j_0+1, \dots, j_0+L-1\}$.
Note that, depending on the values of $N$ and $L$, some atoms may not participate in the formation of the square atom array.
These redundant atoms can be placed at any sites as long as \Cref{cond:grid_reconf} is satisfied.
Therefore, we exclude the target geometry $A^\mathrm{(tgt)}$ from the 4-tuple $(n, A^\textrm{(int)}, V, F)$ that defines the problem instance.

\subsection{Gale--Ryser Theorem}

The Gale--Ryser theorem~\cite{Gale1957-kk,Ryser1957-kk} provides a necessary and sufficient condition for the existence of an atom geometry defined by specified row sums and column sums.
This theorem guarantees our ability to construct a sequence of atom geometries that bridge the gap between the initial geometry $A^\mathrm{(int)}$ and the target geometry $A^\mathrm{(tgt)}$.

For geometry $A \in S$, we define the row sums $R=(r_1, r_2, \dots, r_n)$ and column sums $C = (c_1, c_2, \dots, c_n)$ as follows:
\begin{equation}
     r_i = \sum_{j=1}^{n} A_{i,j}, \quad c_j = \sum_{i=1}^{n} A_{i,j},
\end{equation}
where $r_i$ and $c_j$ denote the sums of the $i$-th row and that of the $j$-th column, respectively.
A basic necessary condition for the existence of an atom geometry $A$ is that the total number of atoms must be consistent between the row sums and column sums:
\begin{equation}
    \sum_{i=1}^{n} r_i = \sum_{j=1}^{n} c_j. \label{eq:condition}
\end{equation}

Here, we consider the inverse problem that constructs an atom geometry $A$ for given row and column sums that satisfy \Cref{eq:condition}.
The Gale--Ryser theorem gives a nontrivial solution to this inverse problem.

\begin{figure}[t]
\begin{algorithm}[H]
    \caption{Atom geometry construction} \label{alg:intermediate}
    \begin{algorithmic}[1]
        \Require Matrix size $n$, row sums $R$, and column sums $C$.
        \Ensure Atom geometry $A$ satisfying constraints on $R$ and $C$.
        \State $A \gets $ zero matrix of size $n$ \Comment{Initialize the geometry}
        \For{$i = 1$ to $n$}
            \State $B_i \gets (r_i, i)$ \Comment{$B$ is a sequence holding $(r_i, i)$ for all $i \in X$}
        \EndFor
        \For{$j = 1$ to $n$}
            \State sort $B$ by non-increasing order of $r_i$
            \For{$k = 1$ to $c_j$}
                \State $(r_i, i) \gets B_k$ 
                \State $A_{i,j} = 1$
                \State $B_k \gets (r_i - 1, i)$
            \EndFor
        \EndFor
        \State \Return $A$
    \end{algorithmic}
\end{algorithm}
\end{figure}

\begin{theorem}[The Gale--Ryser theorem~\cite{Gale1957-kk,Ryser1957-kk}] \label{theorem:gale_ryser}
    Let $R=(r_1, r_2, \dots, r_n)$ and $C = (c_1, c_2, \dots, c_n)$ be two sequences of non-negative integers.
    Let $C^\prime$ denote the sequence obtained by sorting $C$ in non-increasing order.
    A binary matrix $A$ with row sums $R$ and column sums $C$ exists if and only if the following inequality is satisfied:
    \begin{equation} \label{eq:existence_iff}
        \forall k \in X: \sum^{k}_{j = 1} c^\prime_j \leq \sum^{n}_{i=1} \mathrm{min}(r_i, k).
    \end{equation}
\end{theorem}
According to the Gale--Ryser theorem, a valid atom geometry $A$ can be explicitly constructed if \Cref{eq:existence_iff} is satisfied.
\Cref{alg:intermediate} presents a greedy approach that constructs $A$ from the given $R$ and $C$ in $\bigO(n^2 \log n)$ time.
This greedy algorithm is directly derived from the proof in~\cite{Ryser1957-kk}.

\section{Proposed Method}
We propose a planning algorithm that can reconfigure $N = \bigO(n^2)$ atoms in $\bigO(\sqrt{N})$ time.
Formally, this algorithm provides a constructive proof of \Cref{theorem:The_result}.
\begin{theorem}[$\bigO(\sqrt N)$ time arbitrary reconfiguration]\label{theorem:The_result}
    For any problem instance $(n, A^\textrm{(int)}, V, F)$, there exists a feasible plan $p \in P$ satisfying
    \begin{equation}
        F(p, t_1, t_2) = (6n - 6)(t_1 + t_2) \in \bigO(\sqrt{N}).
    \end{equation}
\end{theorem}
The main concept of the proposed algorithm is to exploit maximum parallelism by simultaneously transporting multiple atoms, and to realize this goal, the algorithm adopts a divide-and-conquer approach that reduces reconfiguration time with nearly $\bigO(N)$ times speedup by moving $\bigO(N)$ atoms at once.
As shown in \Cref{fig:overview_proposed}, the proposed algorithm consists of two main components.
\begin{enumerate}
    \item \textbf{Decomposer:} The decomposer breaks the reconfiguration problem into at most three 1D shuttling tasks, such that executing the tasks in sequence reproduces the original problem. As shown in \Cref{fig:1d_shuttling_task}, each 1D shuttling task is a simplified reconfiguration problem in which the target geometry can be achieved by moving atoms along a specific 1D direction.
    \item \textbf{1D Shuttling Solver:} The solver generates a reconfiguration plan for each 1D shuttling task, in which atoms move a distance of at most $(2n - 2) \in\bigO(\sqrt{N})$. The solutions from all tasks are then merged to produce a complete plan for the original problem.
\end{enumerate}

Using \Cref{theorem:gale_ryser}, we show that our divide-and-conquer approach, which generates at most three 1D shuttling tasks, is effective for any reconfiguration problem.
Therefore, the proposed algorithm produces a $3(2n - 2) = (6n-6) \in \bigO(\sqrt{N})$ time plan for any reconfiguration problem (\Cref{theorem:The_result}).
Additionally, we introduce a peephole optimization technique for grid formation problems to achieve higher efficiency in this common scenario.

\begin{figure}[t]
  \centering
  \includegraphics[keepaspectratio, width=.6\linewidth]{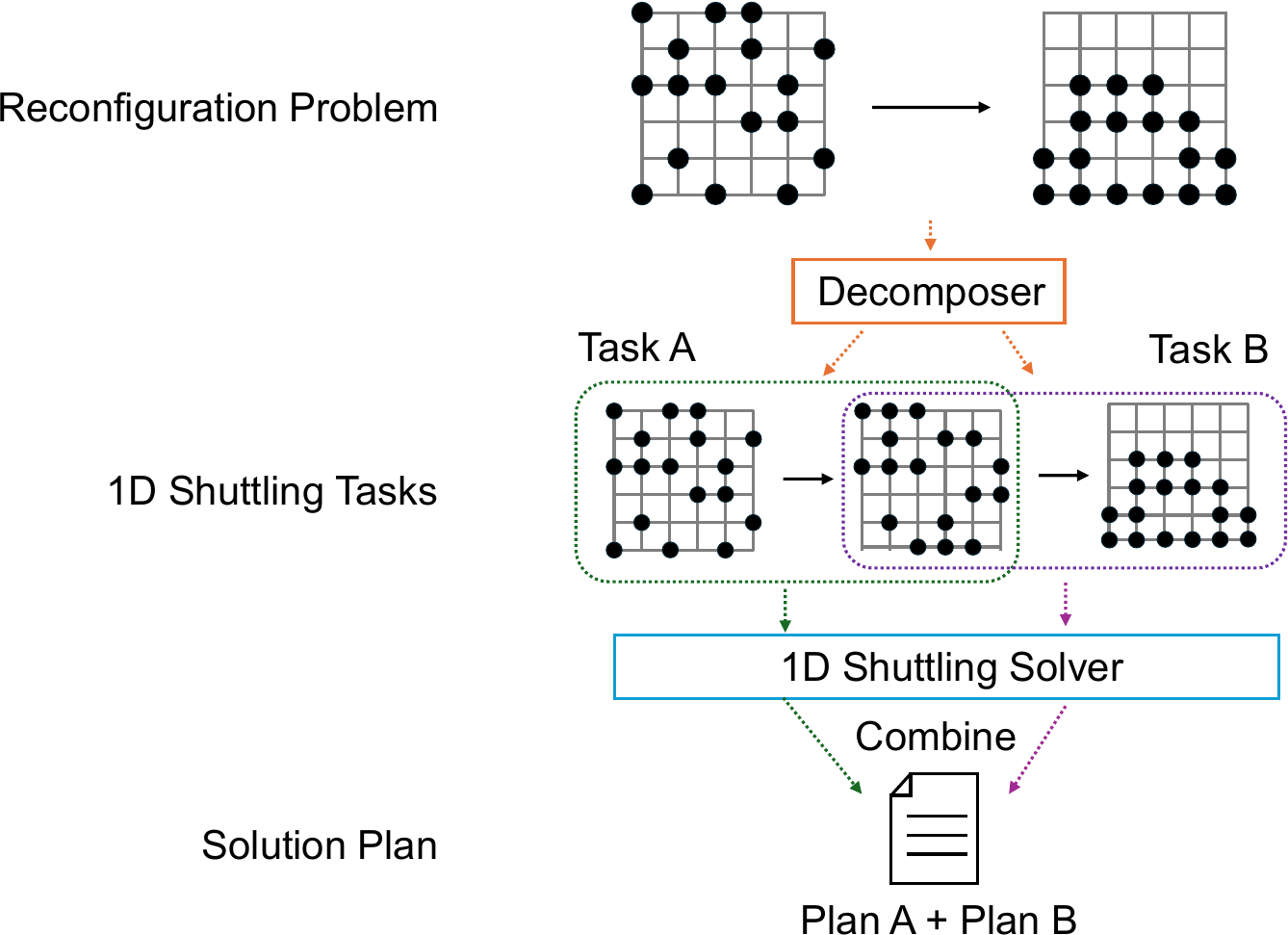}
  \caption{Overview of the proposed algorithm. The problem is first decomposed into multiple 1D shuttling tasks. A solver then processes each 1D shuttling task to obtain local solutions. Finally, the algorithm merges these local solutions to generate a complete reconfiguration plan for the entire problem.}
  \label{fig:overview_proposed}
\end{figure}

\begin{figure}[t]
  \centering
  \begin{subfigure}[t]{0.48\textwidth}
    \centering
    \begin{tabular}{l}
      \subref{subfig:row_wise_shuttling} \\
      \includegraphics[keepaspectratio, width=.9\linewidth]{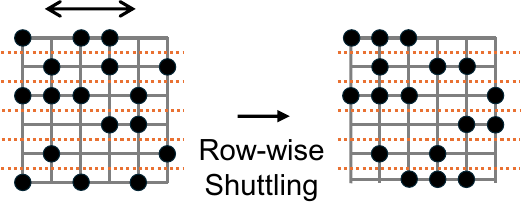}
    \end{tabular}
    \phantomsubcaption\label{subfig:row_wise_shuttling}
 \end{subfigure}
 \begin{subfigure}[t]{0.48\textwidth}
    \centering
    \begin{tabular}{l}
       \subref{subfig:column_wise_shuttling} \\
       \includegraphics[keepaspectratio, width=.9\linewidth]{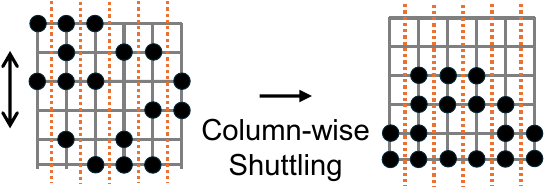}
    \end{tabular}
    \phantomsubcaption\label{subfig:column_wise_shuttling}
 \end{subfigure}
  \caption{Examples of 1D shuttling tasks for a 2D atom array. \subref{subfig:row_wise_shuttling} The row-wise shuttling task uses only left- and right-shift operations to transform the initial configuration into the target configuration. \subref{subfig:column_wise_shuttling} The column-wise shuttling task is analogous, using up- and down-shift operations.}  \label{fig:1d_shuttling_task}
\end{figure}

\begin{figure*}[t]
        \centering
        \includegraphics[keepaspectratio, width=0.97\textwidth]{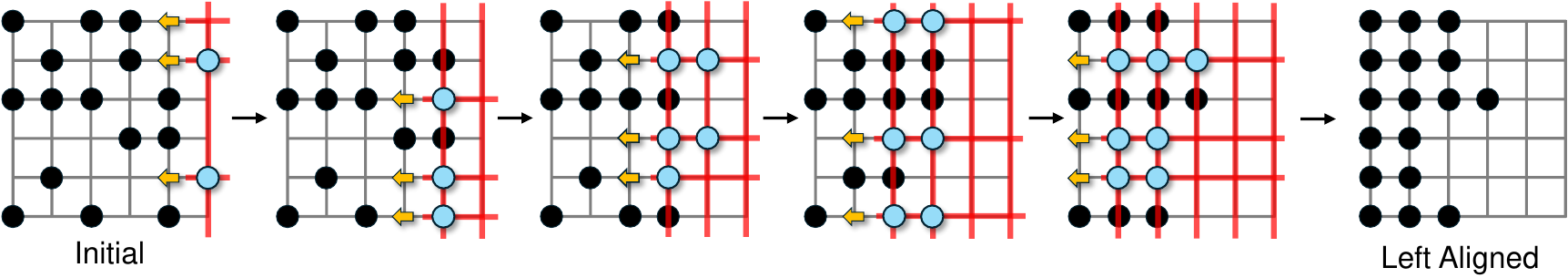}
        \caption{Example of the leftward alignment task. Atoms in the same column are simultaneously shifted leftward from the rightmost column, producing the left-aligned atom configuration.}
        \label{fig:left-align}
\end{figure*}

\begin{figure}[t]
\begin{algorithm}[H]
    \caption{Planning algorithm for the leftward alignment task} \label{alg:left-align}
    \begin{algorithmic}[1]
        \Require The initial geometry $A^\mathrm{(int)}$.
        \Ensure The reconfiguration plan $p$ that converts $A^\mathrm{(int)}$ to the left-aligned geometry $A^\mathrm{(left)}$.
        \State $p \gets \emptyset$ \Comment{Initialize the plan}
        \For{$x = n-1$ downto $1$}
            \State $J \gets \{j \in X \mid x < j \leq n\}$
            \State $I \gets \{i \in X \mid A^\mathrm{(int)}_{i,x} = 0\}$ \Comment{Find empty sites for every row}
            \State Append the $\mathrm{Left}(I,J)$ operation to $p$ \Comment{Left shift}
        \EndFor
        \State \Return $p$
    \end{algorithmic}
\end{algorithm}
\end{figure}

\begin{figure*}[t]
        \centering
        \includegraphics[keepaspectratio, width=0.97\textwidth]{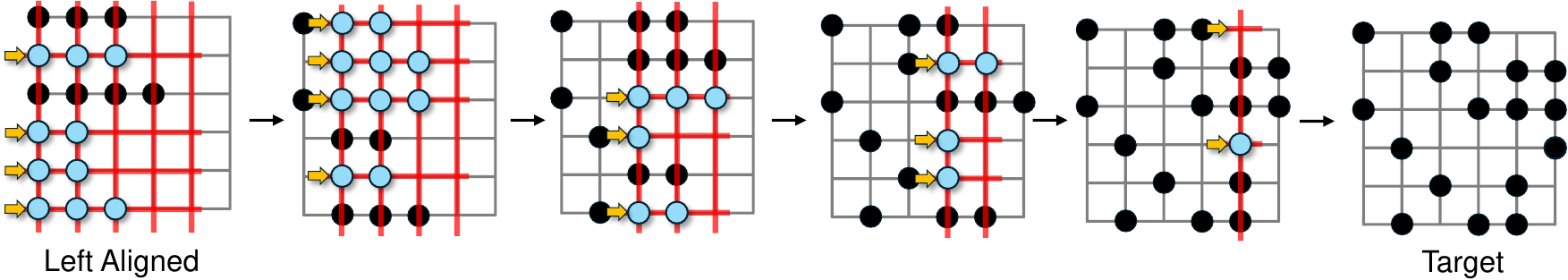}
        \caption{Example of the rightward delivery task. Starting from the left-aligned atom geometry, the selected atoms in each column are shifted rightward to produce the target atom geometry.}
        \label{fig:inverse-left-align}
\end{figure*}

\subsection{1D Shuttling Solver}

We first describe the 1D shuttling solver to facilitate understanding of the decomposer design.
Our solver generates an efficient $(2n-2)$ time plan for an arbitrary 1D shuttling task.
As a specific instance of a 1D shuttling task, we define the row-wise shuttling task, which satisfies the row-sum condition for every row:
\begin{equation}\label{eq:same_row_count}
    \forall i \in X: r^\mathrm{(int)}_{i} = r^\mathrm{(tgt)}_{i},
\end{equation}
where $r^\mathrm{(int)}_{i}$ and $r^\mathrm{(tgt)}_{i}$ are the row sums of the $i$-th row in the initial geometry $A^\mathrm{(int)}$ and the corresponding row in the target geometry $A^\mathrm{(tgt)}$, respectively.
Intuitively, \Cref{eq:same_row_count} requires that each row in $A^\mathrm{(int)}$ and $A^\mathrm{(tgt)}$ contains the same number of atoms, eliminating the need to transport atoms between different rows.
Without loss of generality, we focus only on the row-wise shuttling task because a column-wise shuttling task can be transformed into a row-wise shuttling task by rotating the atom geometry matrix $A$ by 90\textdegree.

The 1D shuttling solver generates a reconfiguration plan consisting of two steps: (1) a leftward alignment step and (2) a rightward delivery step.
Each step is efficiently executed using $(n-1)$ operations by moving multiple atoms simultaneously whenever possible.
Here, we describe the concept and idea of the 1D shuttling solver.
A formal proof of its correctness is provided in \Cref{appendix:1d_shuttling_solver}.

\begin{definition}
A geometry $A \in S$ is said to be \textit{left-aligned} if it satisfies the inequality in \Cref{eq:left-aligned-property}.
\begin{equation}\label{eq:left-aligned-property}
    \forall (i,j) \in X^2: A_{i,j} \geq A_{i, j+1}.
\end{equation}
\end{definition}

\paragraph*{Leftward alignment step (\Cref{fig:left-align}).}
The solver generates the left-aligned geometry $A^\mathrm{(left)}$ by shifting all atoms in $A^\mathrm{(int)}$ to the left. More formally, the left-aligned geometry $A^\mathrm{(left)}$ satisfies \Cref{eq:left-aligned-property} and the following condition:
\begin{equation}\label{constraint:left-align}
\begin{aligned}
    \forall i \in X: \sum^n_{j = 1} A^\textrm{(left)}_{i,j} = \sum^n_{j=1} A^\textrm{(int)}_{i,j}.
\end{aligned}
\end{equation}
Note that $A^\mathrm{(left)}$ can be constructed from an arbitrary initial geometry.
\Cref{alg:left-align} presents the planning algorithm for the leftward alignment step, in which atoms located in the same column are simultaneously shifted leftward starting from the rightmost column.
The key idea is to eliminate empty sites in a given column using a single operation by filling them simultaneously.
To implement this idea, the algorithm inspects columns from right to left, beginning with column $x$ (line~2).
For each column $x$, the algorithm identifies the set of rows $I$ that include an empty site in that column (line~4).
For such rows $I$, the algorithm appends a left-shift operation that removes empty sites in column $x$ by shifting all atoms in columns $x+1, x+2, \dots, n$ leftward.
In other words, more atoms are transported simultaneously as the inspection proceeds from right to left.

The planning time of \Cref{alg:left-align} is $\bigO(n^2) = \bigO(N)$, because the algorithm accesses each element of the atom array at most once.
In contrast, the execution time of the plan generated by \Cref{alg:left-align} is equal to $(n - 1) (t_1 + t_2)$ because the loop at line~2 iterates $n-1$ times and each iteration appends a tweezer operation.

\paragraph*{Rightward delivery step (\Cref{fig:inverse-left-align}).}
The objective of the rightward delivery step is to transform the left-aligned geometry $A^\mathrm{(left)}$ into the target geometry $A^\mathrm{(tgt)}$ by transporting the appropriate atoms rightward.
Clearly, $A^\mathrm{(left)}$ must be the left-aligned geometry of $A^\mathrm{(tgt)}$.
This observation motivates the extension of the left-alignment procedure to implement the rightward delivery step.

\begin{figure}[t]
\begin{algorithm}[H]
    \caption{Planning algorithm for the rightward delivery task} \label{alg:drop-off}
    \begin{algorithmic}[1]
        \Require The target geometry $A^\mathrm{(tgt)}$.
        \Ensure The reconfiguration plan $p$ that converts the left-aligned geometry $A^\mathrm{(left)}$ to $A^\mathrm{(tgt)}$.
        \State $p \gets \emptyset$ \Comment{Initialize the plan}
        \For{$x = 1$ to $n-1$}
            \State $J \gets \{j \in X \mid x \leq j < n\}$
            \State $I \gets \{i \in X \mid A^\mathrm{(tgt)}_{i,x} = 0\}$ \Comment{Create empty sites for every row}
            \State Append the $\textrm{Right}(I,J)$ operation to $p$ \Comment{Right shift}
        \EndFor
        \State \Return $p$
    \end{algorithmic}
\end{algorithm}
\end{figure}

\Cref{alg:drop-off} presents the planning procedure for the rightward delivery step, in which the left-aligned atoms in each column are appropriately selected and simultaneously shifted rightward to achieve the target geometry.
This algorithm follows the same principle as the leftward alignment algorithm, namely that all empty sites in a given column are created simultaneously using a single operation.
To implement this idea, the planning algorithm inspects column $x$ from left to right (line~2).
For each column $x$, it identifies the set of rows $I$ that should contain an empty site (line~4).
For these rows $I$, a right-shift operation is appended, which creates empty sites in column $x$ by shifting all atoms located in columns $x, x+1, \dots, n-1$ to the right.
In contrast to the leftward alignment step, the number of atoms transported simultaneously decreases as the inspection progresses from left to right.

Similar to the leftward alignment algorithm, the planning time of \Cref{alg:drop-off} is $\bigO(n^2) = \bigO(N)$.
Moreover, \Cref{alg:drop-off} produces a plan comprising $n-1$ operations, yielding a reconfiguration time of $(n-1)(t_1+t_2)$.

\begin{lemma}[Alignment time]\label{lemma:packing_algorithm}
    For any configuration $A^\mathrm{(int)} \in S$, alignment in any of the four directions (left, right, up, or down) can be achieved with a feasible plan $p \in P$ whose reconfiguration time is
    \begin{equation}
        F(p, t_1, t_2) = (n - 1) (t_1 + t_2).
    \end{equation}
    This plan $p$ consists of $(n-1)$ tweezer operations.
\end{lemma}

\paragraph*{Peephole optimization.}
Redundant shift operations that do not move any atoms can be omitted to achieve more efficient reconfiguration.
Such redundancy can arise in certain rows during the rightward delivery step.
More specifically, after line~3 of \Cref{alg:drop-off}, a column $j \in X$ can be removed from the set $J$ if all sites $(i,j)$ on $j$ (1) must retain their currently occupied atom or (2) can skip the right-shift operation because the atom delivery for row $i$ has already been completed:
\begin{equation}
    \forall i \in X: A^\mathrm{(tgt)}_{i,j} = 1 ~\lor~ \sum^{j}_{j^\prime=1}A^\mathrm{(tgt)}_{i,j^\prime} = r^\mathrm{(tgt)}_i. \label{eq:peephole}
\end{equation}
Both cases are permitted to maintain the current state of column $j$. \Cref{eq:peephole} can be verified in $\bigO(n)$ time for each of the $(n-1)$ operations, resulting in a total of $\bigO(n^2)=\bigO(N)$ time.
Thus, the additional overhead for this row-wise optimization is asymptotically the same as the time complexity of \Cref{alg:drop-off}.
In other words, inter-row optimization is avoided for row-wise operations because it would require $\bigO(n^3)$ time and increase the overall planning algorithm complexity.

\subsection{Decomposer}

The decomposer breaks any reconfiguration problem into at most three 1D shuttling tasks by attempting three strategies in sequence.
The first two strategies may fail to produce a valid solution for some problems because they aim for greater efficiency in specific cases.
If these strategies fail, the final strategy is applied, which guarantees a valid solution for any problem.
\begin{enumerate}
    \item \textbf{Grid formation strategy.} Applicable only to the grid formation problem, this strategy provides the most efficient solution among the three. The problem is decomposed into two 1D shuttling tasks, allowing the peephole optimization to operate efficiently.
    \item \textbf{Two-step strategy.} The second strategy attempts to decompose the problem into two 1D shuttling tasks.
    As an extension of the previous PSC/Tetris algorithm~\cite{cite:multiple_tweezer_kagome,cite:multiple_tweezer_grid}, this strategy is moderately efficient but has a small risk of failure.
    \item \textbf{Three-step strategy.} The final strategy decomposes the problem into three 1D shuttling tasks. We prove that this three-step strategy always produces a valid plan for any problem.
\end{enumerate}

The strategies described above enable the decomposer to improve both efficiency and reliability for atom reconfiguration problems.
Because a 1D shuttling task requires at most $(2n-2)$ operations, the resulting reconfiguration plan can be executed with an average of $(4n-4)$ operations and a maximum of $(6n-6)$ operations in the worst case.

\begin{figure*}[t]
  \centering
  \includegraphics[keepaspectratio,width=0.97\textwidth]{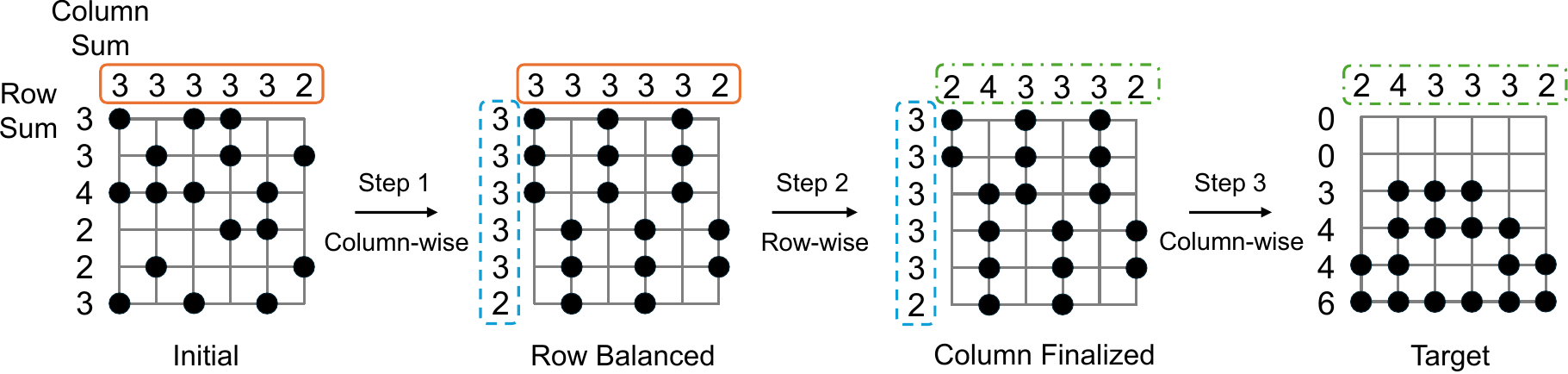}
  \caption{Overview of the three-step strategy. The initial geometry is transformed into the target geometry via two intermediate geometries: the row-balanced geometry and the column-finalized geometry. The transformation proceeds in three steps: (1) a column-wise shuttling task, (2) a row-wise shuttling task, and (3) a final column-wise shuttling task. Colored boxes indicate that either the row sums $R$ or the column sums $C$ remain unchanged between consecutive steps.}
  \label{fig:three-step}
\end{figure*}

\subsubsection{Three-step Strategy}
Given the initial and target atom geometries, $A^\mathrm{(int)}$ and $A^\mathrm{(tgt)}$, the three-step strategy decomposes the reconfiguration problem into the following three 1D shuttling tasks (see \Cref{fig:three-step}).
\begin{enumerate}
    \item \textbf{Row-balancing task.} 
    This task transforms the initial geometry $A^\mathrm{(int)}$ into a row-balanced geometry $A^\mathrm{(rbal)}$ by performing column-wise shuttling to equalize the number of atoms across rows.
    More formally, $A^\mathrm{(rbal)}$ satisfies the following condition:
\begin{equation} \label{eq:row_diff_one}
        \forall i',i'' \in X: |r^\mathrm{(rbal)}_{i'} - r^\mathrm{(rbal)}_{i''}| \leq 1.
\end{equation}
    \item \textbf{Column-finalizing task.}
    The second task transforms $A^\mathrm{(rbal)}$ into a column-finalized geometry $A^\mathrm{(cfin)}$ through row-wise shuttling.
    The geometry $A^\mathrm{(cfin)}$ has the same row sums $R^\mathrm{(cfin)}$ as $A^\mathrm{(rbal)}$ and the same column sums $C^\mathrm{(cfin)}$ as the target geometry $A^\mathrm{(tgt)}$.    
    The existence of such a geometry $A^\mathrm{(cfin)}$ is guaranteed by \Cref{theorem:gale_ryser}.
    \item \textbf{Row-finalizing task.}
    The last task transforms $A^\mathrm{(cfin)}$ into the target geometry $ A^\mathrm{(tgt)}$ by performing column-wise shuttling.
\end{enumerate}

We show that the three-step strategy always produces a valid atom geometry by decomposing any given problem into three 1D shuttling tasks.
This holds if $A^\mathrm{(rbal)}$ and $A^\mathrm{(cfin)}$ exist for any problem.
The existence of $A^\mathrm{(rbal)}$ is ensured by \Cref{alg:leveling}, which constructs $A^\mathrm{(rbal)}$ by distributing atoms across rows using a round-robin approach.
In contrast, the existence of $A^\mathrm{(cfin)}$ is guaranteed by the following \Cref{lemma:strongest}.
\begin{lemma} \label{lemma:strongest}
    Let $A$ be an atom geometry with row sums $R=(r_1, r_2, \dots, r_n)$ satisfying \Cref{eq:row_diff_one}.
    For any sequence of non-negative integers $C^\prime=(c^\prime_1, c^\prime_2, \dots, c^\prime_n)$ that satisfies
    \begin{equation}
        \sum^n_{j=1} c^\prime_{j} = \sum^n_{i=1} r_i,
    \end{equation}
    where $c_j^\prime \in X~(1 \leq j \leq n)$, there exists an atom geometry $A^\prime$ whose row sums are $R$ and column sums are $C^\prime$.
\end{lemma}
\begin{proof}
The details are presented in \Cref{appendix:three-step-strategy}.
\end{proof}

\begin{figure}[t]
\begin{algorithm}[H]
    \caption{Row balancing step} \label{alg:leveling}
    \begin{algorithmic}[1]
        \Require The initial geometry $A^\mathrm{(int)}$.
        \Ensure The row balanced geometry $A^\mathrm{(rbal)}$.
        \State $A^\mathrm{(rbal)} \gets$ zero matrix of size $n \times n$ \Comment{Initialize the geometry}
        \State $t \gets 0$ \Comment{Initialize the atom counter}
        \For{$j = 1$ to $n$}
            \For{$i = 1$ to $n$}
                \If{$A^\textrm{(int)}_{i,j} = 1$}
                    \State $A^\mathrm{(rbal)}_{t+1,j} \gets 1$
                    \State $t \gets (t+1) \mod n$ \Comment{Round-robin distribution}
                \EndIf
            \EndFor
        \EndFor
        \State \Return $A^\mathrm{(rbal)}$
    \end{algorithmic}
\end{algorithm}
\end{figure}

\subsubsection{Two-step Strategy}

The two-step strategy attempts to bypass the row-balancing task used in the three-step strategy, thereby reducing the number of 1D shuttling tasks from three to two.
While we expect the two-step strategy to succeed for most reconfiguration problems, it can occasionally fail to produce a valid plan because the given initial geometry does not always satisfy \Cref{eq:row_diff_one}.

The idea of the two-step strategy is to transform the initial geometry $A^\mathrm{(int)}$ into either a column-finalized geometry $A^\mathrm{(cfin)}$ or a row-finalized geometry $A^\mathrm{(rfin)}$ that can subsequently be converted into the target geometry $A^\mathrm{(tgt)}$.
If a column-finalized geometry is obtained, the two-step strategy performs a row-finalizing task, which is similar to the approach used in the PSC/Tetris algorithms.
If a row-finalized geometry is obtained instead, the two-step strategy performs a column-finalizing task, which is not included in the PSC/Tetris algorithms.
If neither geometry can be found, the algorithm switches to the three-step strategy to ensure a valid reconfiguration plan.
This hybrid approach exploits the higher efficiency of the two-step strategy whenever possible, while relying on the three-step strategy to guarantee a solution for all problem instances. 

The search algorithm for finding $A^\mathrm{(cfin)}$ or $A^\mathrm{(rfin)}$ relies on verifying the row sums $R$ and column sums $C$.
Specifically, given the initial geometry $A^\mathrm{(int)}$ and the target geometry $A^\mathrm{(tgt)}$, we try to find $A^\mathrm{(cfin)}$ such that
\begin{equation}
    R^\mathrm{(cfin)} = R^\mathrm{(int)} ~\wedge~ C^\mathrm{(cfin)} = C^\mathrm{(tgt)}.
\end{equation}
Similarly, we try to find $A^\mathrm{(rfin)}$ such that
\begin{equation}
    R^\mathrm{(rfin)} = R^\mathrm{(tgt)} ~\wedge~ C^\mathrm{(rfin)} = C^\mathrm{(int)}.
\end{equation}
In both cases, the existence of a valid geometry can be verified using \Cref{theorem:gale_ryser}.
For example, consider the following initial and target geometries, which satisfy \Cref{eq:existence_iff}:
\begin{align}
    A^\mathrm{(int)} &= \begin{pmatrix}
        1& 1 & 1 & 1 \\
        0& 0 & 0 & 0 \\
        0& 0 & 0 & 0 \\
        0& 0 & 0 & 0 \\
    \end{pmatrix},\\
    A^\mathrm{(tgt)} &= \begin{pmatrix}
        1& 0 & 0 & 0 \\
        1& 0 & 0 & 0 \\
        1& 0 & 0 & 0 \\
        1& 0 & 0 & 0 \\
    \end{pmatrix}.
\end{align}
Clearly, $A^\mathrm{(cfin)}$ does not exist in this example because both $R^\mathrm{(int)}$ and $C^\mathrm{(tgt)}$ are $(4,0,0,0)$.
In contrast, $C^\mathrm{(int)}=R^\mathrm{(tgt)}=(1,1,1,1)$ gives $R^\mathrm{(rfin)}=C^\mathrm{(rfin)}=(1,1,1,1)$, which satisfies \Cref{eq:existence_iff} and allows \Cref{alg:intermediate} to construct $A^\mathrm{(rfin)}$ as follows.
\begin{equation}
    A^\mathrm{(rfin)} = \begin{pmatrix}
        1& 0 & 0 & 0 \\
        0& 1 & 0 & 0 \\
        0& 0 & 1 & 0 \\
        0& 0 & 0 & 1 \\
    \end{pmatrix}.
\end{equation}

No valid geometry exists if \Cref{eq:existence_iff} is not satisfied.
For example, consider the following initial and target geometries.
\begin{align}
    A^\mathrm{(int)} &= \begin{pmatrix}
        1& 1 & 1 & 1 \\
        1& 1 & 1 & 0 \\
        1& 0 & 0 & 0 \\
        1& 0 & 0 & 0 \\
    \end{pmatrix},\\
    A^\mathrm{(tgt)} &= \begin{pmatrix}
        1 & 1 & 1 & 0 \\
        1 & 1 & 1 & 0 \\
        1 & 1 & 1 & 0  \\
        0 & 0 & 0 & 0 \\
    \end{pmatrix}.
\end{align}
For $A^\mathrm{(cfin)}$, we require $R^\mathrm{(cfin)}=(4,3,1,1)$ and $C^\mathrm{(cfin)}=(3,3,3,0)$, which do not satisfy \Cref{eq:existence_iff}.
Similarly, for $A^\mathrm{(rfin)}$, we require $R^\mathrm{(rfin)}=(3,3,3,0)$ and $C^\mathrm{(rfin)}=(4,2,2,1)$, which also fail to satisfy \Cref{eq:existence_iff}.

\subsubsection{Grid Formation Strategy}
The grid formation strategy, which aims to create an $L \times L$ atom array, attempts to decompose the original reconfiguration problem into two 1D shuttling tasks in a way that allows the peephole optimization technique to further reduce the number of required tweezer operations, with the key idea being to gather atoms toward the top-left corner region of the atom array.

\begin{figure*}[t]
  \centering
  \includegraphics[keepaspectratio,width=0.80\textwidth]{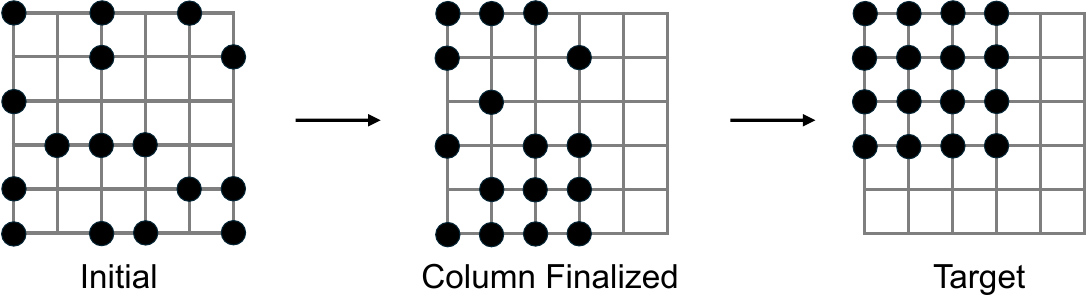}
  \caption{Illustration of the grid formation procedure.}
  \label{fig:staircase}
\end{figure*}

Given the initial geometry $A^\mathrm{(int)}$, the grid formation strategy decomposes the grid formation problem into two sequential tasks, as shown in \Cref{fig:staircase}.
\begin{enumerate}
    \item \textbf{Column-finalizing task.} This task transforms the initial geometry $A^\mathrm{(int)}$ into a column-finalized geometry $A^\mathrm{(cfin)}$ by performing row-wise shuttling to gather $L^2$ atoms into the leftmost $n \times L$ sites.
    After applying the peephole optimization, this task requires $(n-1) + (L-1)$ operations.
    \item \textbf{Row-finalizing task.} This task transforms $A^\mathrm{(cfin)}$ into the target geometry $A^\mathrm{(tgt)}$ by performing column-wise shuttling to gather $L^2$ atoms to the top-left $L \times L$ sites.
    After applying the peephole optimization, this task requires $(n-1)$ operations.
\end{enumerate}
As a result, the grid formation strategy produces a reconfiguration plan that contains $2(n-1) + (L - 1)$ tweezer operations, which is approximately less than half the number of operations required by the three-step reconfiguration strategy.

\begin{figure}[t]
\begin{algorithm}[H]
    \caption{Column finalizing step for grid formation} \label{alg:grid_int}
    \begin{algorithmic}[1]
        \Require The initial geometry $A^\textrm{(int)}$.
        \Ensure The column finalized geometry $A^\textrm{(cfin)}$.
        \State $A^\textrm{(cfin)} \gets$ zero matrix of size $n \times n$ \Comment{Initialize the geometry}
        \State $t \gets 0$ \Comment{Initialize the atom counter}
        \For{$i = 1$ to $n$}
            \State $l \gets \mathrm{min}(\sum^{n}_{j=1} A^\textrm{(int)}_{i,j}, L)$ \Comment{At most $L$ atoms per row}
            \For{$j = 1$ to $l$}
                \State $A^\textrm{(cfin)}_{i, t+1} \gets 1$
                \State $t \gets (t+1) \mod L$ \Comment{Round-robin distribution to $L$ columns}
            \EndFor
            \For{$j = l + 1$ to $\sum^{n}_{k=1} A^\textrm{(int)}_{i,k}$} \Comment{Remaining atoms are placed outside the $L$ columns}
                \State $A^\textrm{(cfin)}_{i,j} \gets 1$
            \EndFor
        \EndFor
        \State \Return $A^\textrm{(cfin)}$
    \end{algorithmic}
\end{algorithm}
\end{figure}

The column-finalized geometry $A^\mathrm{(cfin)}$ can be constructed using a round-robin procedure, as described in \Cref{alg:grid_int}.
For each row $i$, the algorithm first computes the number $l$ of atoms that can be allocated to the $L \times L$ atom array (line~4).
It then distributes these $l$ atoms across the $L$ columns in a round-robin manner (lines~5--8), while any remaining atoms are placed outside the $n \times L$ region (lines~9--11).
Given the column-finalized geometry $A^\mathrm{(cfin)}$, the target geometry $A^\mathrm{(tgt)}$ can be easily obtained by performing column-wise packing because the atoms in $A^\mathrm{(cfin)}$ are evenly distributed among the $L$ columns.

During the column-finalizing task, if a row initially contains more than $L$ atoms, the strategy discards the excess atoms even though they could potentially be used to compensate for rows with fewer atoms.
As a result, the grid formation strategy fails to construct the target $L \times L$ grid if the number of retained atoms falls below $L^2$.
More formally, the grid formation strategy produces a valid reconfiguration plan if the initial geometry $A^\mathrm{(int)}$ satisfies the following condition:
 \begin{equation}
    \sum^{n}_{i=1} \mathrm{min}\left(\sum^{n}_{j=1}A^\mathrm{(int)}_{i,j}, L\right) \geq L^2.
\end{equation}
Otherwise, the algorithm must switch to the three-step reconfiguration strategy to produce a valid plan.

\section{Related Work}

\Cref{tb:comparison} compares the proposed algorithm with previous state-of-the-art (SOTA) algorithms.
Many existing planning algorithms assume the use of a single mobile tweezer that transports one atom at a time; for example, the Hungarian algorithm~\cite{cite:Hungarian_original,cite:Hungarian} computes an optimal plan by formulating the atom reconfiguration problem as a graph matching problem.
Similar optimization-based approaches have been proposed previously~\cite{PhysRevA.102.063107,cite:hardware_AOD_assembler}.
These single-tweezer algorithms were later enhanced by multitweezer algorithms~\cite{cite:first_experiment_on_AOD,cite:multiple_tweezer_kagome}, which allow multiple atoms to move simultaneously.
However, these earlier algorithms require $N$ control devices to move $N$ atoms, limiting scalability owing to hardware costs that increase linearly with the number of atoms.

This linear cost has been mitigated by the development of a 2D AOD tweezer setup~\cite{cite:crossed_AOD}, in which AODs are arranged in a crossed configuration to control multiple mobile tweezers in parallel.
However, this approach constrains the mobile tweezers to form a 2D lattice pattern, thereby eliminating the ability to independently control the positions of individual tweezers.
As a result, planning algorithms must generate carefully designed reconfiguration plans to fully exploit the available parallel transportation capability.

\renewcommand{\arraystretch}{1.2}
\begin{table*}[t]
  \centering
  \setlength{\tabcolsep}{.2em}
  \caption{Comparison of atom reconfiguration algorithms.}
  \footnotesize
  \begin{tabular}{@{}lccccc@{}}
    \hline
    Item & Proposed & PSC~\cite{cite:multiple_tweezer_kagome}/Tetris~\cite{cite:multiple_tweezer_grid} & Hungarian~\cite{cite:Hungarian,cite:Hungarian_original} & Zhu \cite{Zhu2022-ga} & Constantinides \cite{cite:optimal_routing}\\ \hline
    Target array geometry  & Arbitrary & Arbitrary & Arbitrary & Grid & Arbitrary \\
    Distinguish atoms & No & No & No & No & Yes \\
    Number of AODs & 2 & 2 & 2 & 2 & 2\\
    Selective transfer required & No & No & No & Yes & Yes\\
    Total transportation cost & $\bigO (\sqrt N)$ & $\bigO(N)$ & $\bigO (N)$ & $\bigO(N \log \sqrt N)$ & Unknown \\
    Number of operations  & $\bigO (\sqrt N)$  & $\bigO (\sqrt N)$ & $\bigO(N)$ & $\bigO (\sqrt N)$ & $\bigO(\sqrt N \log \sqrt N)$ \\
    Planning time & $\bigO (N \log \sqrt N)$ & $\bigO(N)$ & $\bigO(N^4)$ &$\bigO(N)$ &  Unknown \\
    \hline
  \end{tabular}
  \label{tb:comparison}
\end{table*}
\renewcommand{\arraystretch}{1.0}

The PSC algorithm~\cite{cite:multiple_tweezer_kagome} is a well-known algorithm that leverages the capabilities of 2D AOD tweezers.
The Tetris algorithm, which is identical to the PSC algorithm, was independently proposed by Wang et al.~\cite{cite:multiple_tweezer_grid}.
Both algorithms reconfigure atoms into arbitrary target geometries through two steps:
(1) row-by-row rearrangement to construct Tetrimino blocks, and (2) column-by-column compression to eliminate the Tetrimino blocks and obtain the target atom geometry.
This two-step strategy, which transports $\bigO(\sqrt N)$ atoms at a time, reconfigures $N$ atoms in $\bigO(t_1 \sqrt N + t_2 N)$ time under the linear cost model.
The proposed algorithm can be viewed as a refined variant of the PSC algorithm because both algorithms move multiple atoms using row-wise and column-wise operations.
However, the proposed algorithm differs from the PSC algorithm in that it (1) achieves an $\bigO(\sqrt N)$ total transportation cost for scalable neutral-atom systems, (2) guarantees a valid plan for arbitrary target geometries, and (3) incorporates peephole optimization for grid target geometries.

Exploiting transportation parallelism becomes easier when the underlying hardware supports selective transfer, in which a subset of atoms can be chosen to remain stationary while others are transported.
Zhu et al.~\cite{Zhu2022-ga} proposed an atom reconfiguration algorithm that exploits the full capability of 2D AOD devices through selective transfer.
However, its worst-case time complexity for atom reconfiguration is $\bigO(t_1 \sqrt N + t_2 N)$, which is the same as that of the PSC algorithm.
Constantinides et al.~\cite{cite:optimal_routing} further leverage selective transfer to achieve an atom reconfiguration time of $\bigO((t_1 \sqrt N + t_2 N) \log \sqrt N)$ for more constrained problems in which each atom has a specific destination site.
Such strict destination constraints are not considered in our target problem setting, where atoms are interchangeable, and any atom may occupy any site as long as the final geometry matches the target geometry.
In contrast, the proposed algorithm focuses on constructing atom arrays without relying on selective transfer, thereby easing hardware requirements and facilitating practical hardware development.

\section{Evaluation}
\label{sec:evaluation}

\begin{table}[t]
 \caption{Environment for numerical experiments.}
 \label{table:env-numerical}
 \centering
  \begin{tabular}{ll}
   \hline
   Item & Specification\\
   \hline
   CPU & Intel Core i5-12400F \\
   RAM & 16GiB \\
   OS & Ubuntu 24.04.3 LTS \\
   g++ version & 13.3.0 \\
   Compile options & \verb|-O2 --std=c++17| \\
   \hline
  \end{tabular}
\end{table}

\begin{figure}[t]
  \centering
  \begin{subfigure}[t]{0.46\textwidth}
    \centering
    \begin{tabular}{l}
      \subref{subfig:grid_movement_cost} \\
      \includegraphics[keepaspectratio, width=\linewidth]{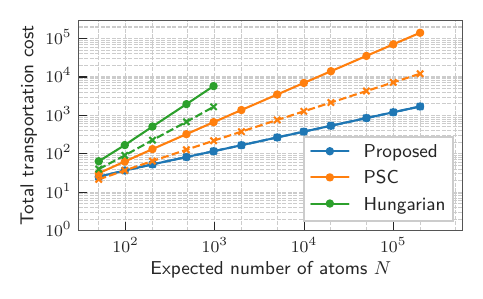}
    \end{tabular}
    \phantomsubcaption\label{subfig:grid_movement_cost}
  \end{subfigure}
    \begin{subfigure}[t]{0.46\textwidth}
    \centering
    \begin{tabular}{l}
      \subref{subfig:arbitrary_movement_cost} \\
      \includegraphics[keepaspectratio, width=\linewidth]{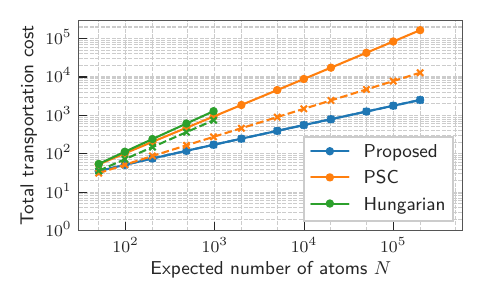}
    \end{tabular}
    \phantomsubcaption\label{subfig:arbitrary_movement_cost}
  \end{subfigure}
  \begin{subfigure}[t]{0.46\textwidth}
    \centering
    \begin{tabular}{l}
       \subref{subfig:grid_op_count} \\
       \includegraphics[keepaspectratio, width=\linewidth]{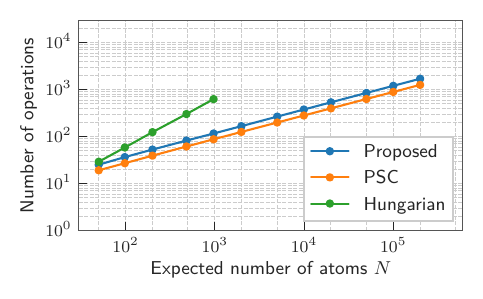}
    \end{tabular}
    \phantomsubcaption\label{subfig:grid_op_count}
  \end{subfigure}
    \begin{subfigure}[t]{0.46\textwidth}
    \centering
    \begin{tabular}{l}
       \subref{subfig:arbitrary_op_count} \\
       \includegraphics[keepaspectratio, width=\linewidth]{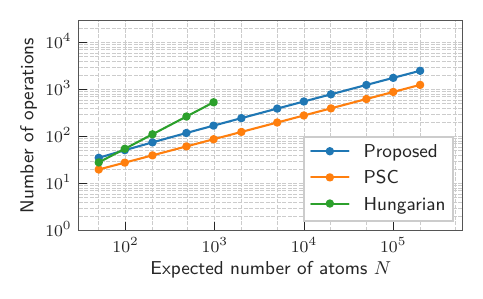}
    \end{tabular}
    \phantomsubcaption\label{subfig:arbitrary_op_count}
  \end{subfigure}
  \begin{subfigure}[t]{0.46\textwidth}
    \centering
    \begin{tabular}{l}
      \subref{subfig:grid_reconfiguration_time}\\
      \includegraphics[keepaspectratio, width=\linewidth]{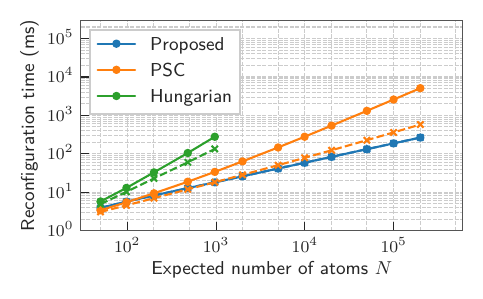}
    \end{tabular} \phantomsubcaption\label{subfig:grid_reconfiguration_time}
  \end{subfigure}
    \begin{subfigure}[t]{0.46\textwidth}
    \centering
    \begin{tabular}{l}
      \subref{subfig:arbitrary_reconfiguration_time}\\
      \includegraphics[keepaspectratio, width=\linewidth]{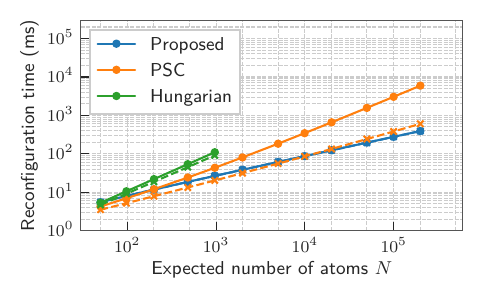}
    \end{tabular}
    \phantomsubcaption\label{subfig:arbitrary_reconfiguration_time}
  \end{subfigure}
  \caption{Evaluation results for reconfiguration plans.
  Total transportation cost for \subref{subfig:grid_movement_cost} grid and \subref{subfig:arbitrary_movement_cost} arbitrary target geometries.
  Number of operations for \subref{subfig:grid_op_count} grid and \subref{subfig:arbitrary_op_count} arbitrary target geometries.
  Estimated reconfiguration time for \subref{subfig:grid_reconfiguration_time} grid and \subref{subfig:arbitrary_reconfiguration_time} arbitrary target geometries.
  Solid and dashed lines represent results for the linear and sqrt cost models, respectively.} \label{fig:grid}
\end{figure}

We evaluated the efficiency and scalability of the proposed algorithm through numerical experiments, which
empirically confirm the theoretical $\mathcal{O}(\sqrt{N})$ time complexity and demonstrate the benefits of exploiting the inherent parallelism of tweezer arrays combined with a lattice pattern.
The evaluation provides a comprehensive analysis consisting of two main parts.
\begin{itemize}
    \item \textbf{Evaluation for reconfiguration plans.} We compared the reconfiguration plans generated by the proposed algorithm with those produced by SOTA algorithms~\cite{cite:multiple_tweezer_kagome,cite:multiple_tweezer_grid,cite:Hungarian}.
    This evaluation demonstrates that the proposed algorithm effectively exploits hardware parallelism to achieve performance improvements over previous algorithms.
    We also performed a detailed comparison with the PSC/Tetris algorithm~\cite{cite:multiple_tweezer_kagome,cite:multiple_tweezer_grid}.
    \item \textbf{Evaluation for planning algorithms.} We compared the proposed and SOTA algorithms in terms of planning time and reliability.
\end{itemize}
As SOTA baselines, we considered two algorithms: (1) the PSC/Tetris algorithm~\cite{cite:multiple_tweezer_kagome,cite:multiple_tweezer_grid}, a heuristic algorithm designed for 2D AOD tweezer systems; and (2) the Hungarian algorithm~\cite{cite:Hungarian}, a graph-based approach that minimizes total travel distance through optimal atom-to-site assignment.
Note that all previous algorithms were evaluated using the complex cost model described in \Cref{appendix:difficult_operations}, or their respective original models.
The Hungarian algorithm provides a theoretical lower bound on the total travel distance achievable when using a single mobile tweezer system.

The numerical experiments were performed for two problem types (grid formation and arbitrary formation) using the setup described in \Cref{table:env-numerical}.
The reported results are averages over 10,000 randomly generated problem instances.

\subsection{Evaluation of Reconfiguration Plan}

\Cref{fig:grid} shows the evaluation results for the reconfiguration plans generated by the SOTA algorithms.
All experiments were performed with a time limit of 10 h per algorithm to obtain reconfiguration plans.
Under this time constraint, we were unable to obtain large-scale results ($N > 10^3$) for the Hungarian algorithm.
The proposed algorithm yields identical results under both the linear and sqrt cost models because each operation moves atoms with a unit travel distance of 1: $1=\sqrt1$.

In \Cref{subfig:grid_movement_cost}, the proposed algorithm achieves the lowest total transportation cost for grid geometries when $N \geq 200$.
In particular, compared with the PSC algorithm, the proposed method reduces the total transportation cost to approximately $1/83$ and $1/7$ under the linear and sqrt cost models, respectively, when $N = 2 \times 10^5$.
The smaller reduction under the sqrt cost model arises because the PSC algorithm benefits from long-distance transportation, which allows atoms to be moved farther in a single operation.
Moreover, the experimental results are consistent with the theoretical analysis, showing that the total transportation cost of the proposed algorithm scales with $\sqrt N$ and exhibits excellent asymptotic behavior.
Indeed, the scaling exponent of the proposed algorithm is the lowest among all compared algorithms, regardless of the cost models used.

The performance trends for arbitrary geometries shown in \Cref{subfig:arbitrary_movement_cost} are similar to those observed for grid geometries, with the proposed algorithm achieving the lowest total transportation cost among the compared methods when $N \geq 200$.
In particular, compared with the PSC algorithm, the proposed algorithm reduces the total transportation cost to approximately $1/66$ and $1/5$ under the linear and sqrt cost models, respectively, when $N = 2 \times 10^5$.
Likewise, the Hungarian algorithm also achieves a reduction in total transportation cost for arbitrary geometries, primarily because it allows long-distance transportation of atoms.

With respect to the number of operations, the proposed algorithm achieves competitive performance compared to the PSC algorithm, as shown in \Cref{subfig:grid_op_count}, with both algorithms exhibiting an $\mathcal{O}(\sqrt{N})$ scaling trend for grid geometries, which is superior to that of the Hungarian algorithm.
The proposed algorithm incurs a modest increase in the number of operations, requiring approximately $32$\%--$35$\% more operations than the PSC algorithm.
For arbitrary geometries, \Cref{subfig:arbitrary_op_count} shows that the proposed algorithm requires approximately $77$\%--$100$\% more operations than the PSC algorithm, and the larger increase from grid to arbitrary geometries highlights the effect of the peephole optimization, which effectively reduces the number of operations for grid geometries but provides less benefit for arbitrary geometries.
Although the reconfiguration plans for arbitrary formation involve more operations than those for grid formation, the observed results are consistent with the theoretical analysis, which predicts $\bigO(n)$ scaling behavior.

\Cref{subfig:grid_reconfiguration_time} shows the estimated reconfiguration time for grid geometries using $t_1 = 120~\si{\micro\second}$ and $t_2 = 35~\si{\micro\second}$, which are the same parameter values used in reference \cite{cite:multiple_tweezer_kagome}.
The results show that the proposed algorithm achieves the shortest estimated reconfiguration time among the compared algorithms when $N \geq 2000$.
These findings indicate that for large-scale atom arrays, the advantage of reducing the total transportation cost outweighs the drawback of an increased number of operations.
Moreover, \Cref{subfig:grid_reconfiguration_time} shows that the performance gap between the proposed algorithm and the other methods becomes smaller when the linear cost model is replaced with the sqrt cost model because this replacement enhances the benefit of long-distance transportation, while the proposed algorithm assumes identical transportation costs under both models.
Overall, these results empirically confirm that the proposed algorithm achieves $\bigO(\sqrt{N})$ scaling for reconfiguration time, in agreement with the theoretical analysis.

\Cref{subfig:arbitrary_reconfiguration_time} presents the estimated reconfiguration time for arbitrary geometries using the same parameter settings as those for grid geometries.
Owing to the increased number of operations, the estimated reconfiguration time is longer than that for grid formation problems.
For example, the estimated reconfiguration time for arbitrary geometries is approximately 1.40--1.48$\times$ longer than that for grid geometries.
However, the proposed algorithm maintained its scalability advantage over comparable algorithms, even under the sqrt cost model.
The observed performance gap between grid and arbitrary geometries arises from differences in the atom density of the target geometry; grid geometries exhibit higher local density because atoms are compacted into a specific square region.
The proposed algorithm benefits from this higher density because the target grid geometry structurally resembles the left-aligned geometry, allowing a substantial portion of the rightward delivery step to be skipped through peephole optimization, which results in approximately a 33\% reduction in reconfiguration time. 

\begin{figure}[t]
  \centering
  \begin{subfigure}[t]{0.48\textwidth}
    \centering
    \begin{tabular}{l}
       \subref{subfig:atom_moved_parallelism}\\
       \includegraphics[keepaspectratio, width=\linewidth]{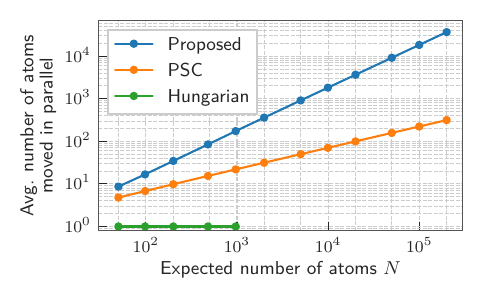}
    \end{tabular}
    \phantomsubcaption\label{subfig:atom_moved_parallelism}
  \end{subfigure}
  \begin{subfigure}[t]{0.48\textwidth}
    \centering
    \begin{tabular}{l}
      \subref{subfig:atom_distance}\\
      \includegraphics[keepaspectratio, width=\linewidth]{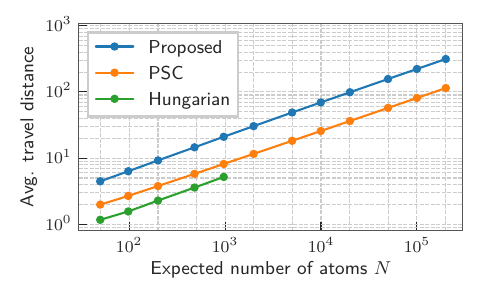}
    \end{tabular}
    \phantomsubcaption\label{subfig:atom_distance}
  \end{subfigure}
  \begin{subfigure}[t]{0.48\textwidth}
    \centering
    \begin{tabular}{l}
      \subref{subfig:grid_avg_moved_atoms}\\
      \includegraphics[keepaspectratio, width=\linewidth]{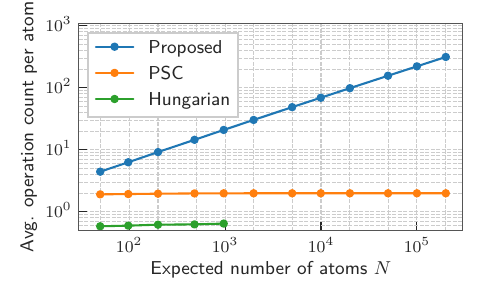}
    \end{tabular}
    \phantomsubcaption\label{subfig:grid_avg_moved_atoms}
  \end{subfigure}
  \caption{Analysis of grid formation plans.
  \subref{subfig:atom_moved_parallelism} Average number of atoms moved in parallel.
  \subref{subfig:atom_distance} Average travel distance per atom.
  \subref{subfig:grid_avg_moved_atoms} Average number of operations per atom.}\label{fig:grid_rowwise_reconf_time}
\end{figure}

\begin{figure}[t]
  \centering
  \begin{subfigure}[t]{0.48\textwidth}
    \centering
    \begin{tabular}{l}
       \subref{subfig:gather_rowwise_reconf_time} \\
       \includegraphics[keepaspectratio, width=\linewidth]{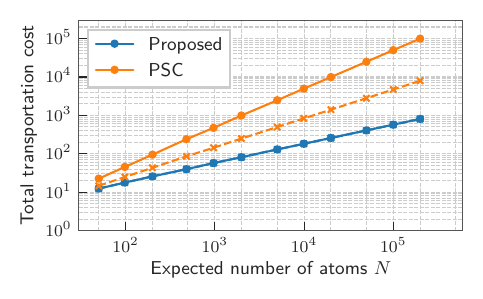}
    \end{tabular}
    \phantomsubcaption\label{subfig:gather_rowwise_reconf_time}
  \end{subfigure}
  \begin{subfigure}[t]{0.48\textwidth}
    \centering
    \begin{tabular}{l}
      \subref{subfig:random_rowwise_reconf_time}\\
      \includegraphics[keepaspectratio, width=\linewidth]{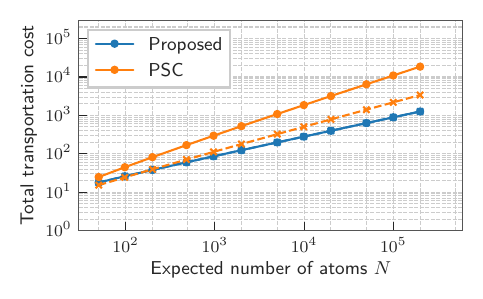}
    \end{tabular}
    \phantomsubcaption\label{subfig:random_rowwise_reconf_time}
  \end{subfigure}
  \caption{Total transportation cost for row-wise operations. \subref{subfig:gather_rowwise_reconf_time} Results for gathering tasks.
  \subref{subfig:random_rowwise_reconf_time} Results for randomizing tasks. Solid and dashed lines represent results for the linear and sqrt cost models, respectively.}  \label{fig:rowwise_reconf_time}
\end{figure}

\begin{figure}[t]
  \centering
  \begin{subfigure}[t]{0.48\textwidth}
    \centering
    \begin{tabular}{l}
       \subref{subfig:grid_planning_time} \\
       \includegraphics[keepaspectratio, width=\linewidth]{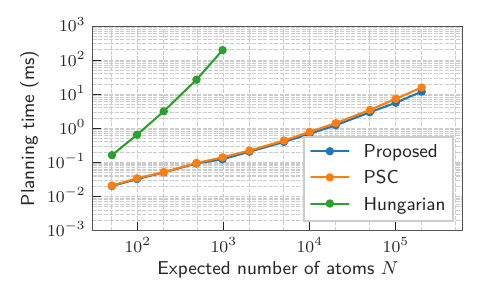}
    \end{tabular}
    \phantomsubcaption\label{subfig:grid_planning_time}
  \end{subfigure}
  \begin{subfigure}[t]{0.48\textwidth}
    \centering
    \begin{tabular}{l}
       \subref{subfig:arbitrary_planning_time} \\
       \includegraphics[keepaspectratio, width=\linewidth]{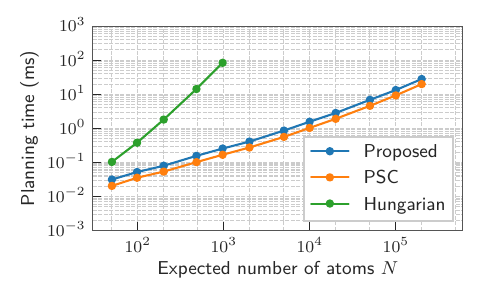}
    \end{tabular}
    \phantomsubcaption\label{subfig:arbitrary_planning_time}
  \end{subfigure}
  \begin{subfigure}[t]{0.48\textwidth}
    \centering
    \begin{tabular}{l}
      \subref{subfig:grid_success_rate}\\
      \includegraphics[keepaspectratio, width=\linewidth]{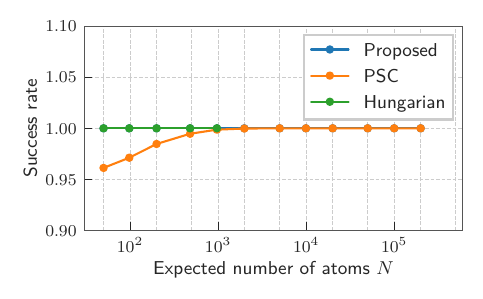}
    \end{tabular}
    \phantomsubcaption\label{subfig:grid_success_rate}
  \end{subfigure}
  \begin{subfigure}[t]{0.48\textwidth}
    \centering
    \begin{tabular}{l}
      \subref{subfig:arbitrary_success_rate}\\
      \includegraphics[keepaspectratio, width=\linewidth]{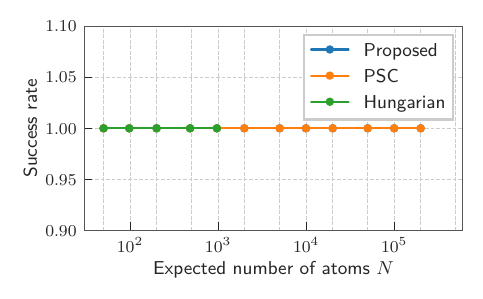}
    \end{tabular}
    \phantomsubcaption\label{subfig:arbitrary_success_rate}
  \end{subfigure}
  \caption{Planning time and success rate for the grid and arbitrary formation problems.
  Planning time for \subref{subfig:grid_planning_time} grid and \subref{subfig:arbitrary_planning_time} arbitrary target geometries.
  Success rate for \subref{subfig:grid_success_rate} grid and \subref{subfig:arbitrary_success_rate} arbitrary target geometries.}  \label{fig:arbitrary_plan}
\end{figure}

\begin{figure}[t]
  \centering
  \begin{subfigure}[t]{0.48\textwidth}
    \centering
    \begin{tabular}{l}
      \subref{subfig:grid_reconf_three_step_ratio}\\
      \includegraphics[keepaspectratio, width=\linewidth]{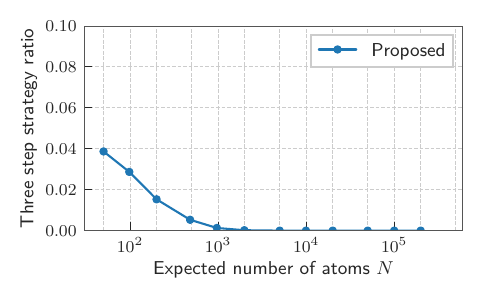}
    \end{tabular}
    \phantomsubcaption\label{subfig:grid_reconf_three_step_ratio}
  \end{subfigure}
  \begin{subfigure}[t]{0.48\textwidth}
    \centering
    \begin{tabular}{l}
       \subref{subfig:arbitrary_reconf_three_step_ratio} \\
       \includegraphics[keepaspectratio, width=\linewidth]{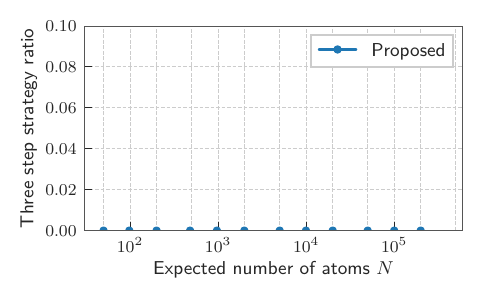}
    \end{tabular}
    \phantomsubcaption\label{subfig:arbitrary_reconf_three_step_ratio}
  \end{subfigure}
  \begin{subfigure}[t]{0.48\textwidth}
    \centering
    \begin{tabular}{l}
       \subref{subfig:grid_peephole_optimization} \\
       \includegraphics[keepaspectratio, width=\linewidth]{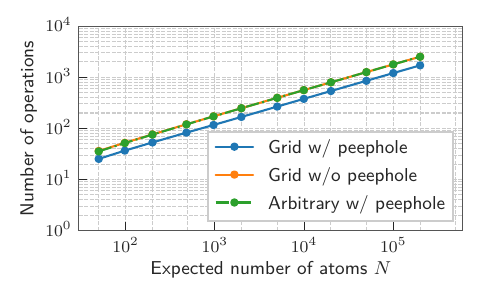}
    \end{tabular}
    \phantomsubcaption\label{subfig:grid_peephole_optimization}
  \end{subfigure}
  \caption{Impact of the three-step strategy and peephole optimization.
  Activation ratio of the three-step strategy for \subref{subfig:grid_reconf_three_step_ratio} grid and \subref{subfig:arbitrary_reconf_three_step_ratio} arbitrary target geometries.
  \subref{subfig:grid_peephole_optimization} Number of operations for grid and arbitrary geometries in the proposed algorithm. Blue and orange lines correspond to grid target geometries with and without peephole optimization, respectively, while the green line represents arbitrary target geometries.}  \label{fig:inner_analysis}
  \end{figure}

\paragraph{Detailed analysis of grid formation plans.}

We next performed a detailed analysis of grid formation plans.
\Cref{subfig:atom_moved_parallelism} shows the number of atoms transported simultaneously, highlighting the superior parallelism achieved by the proposed algorithm.
As illustrated, each operation in the proposed algorithm moves 1.8--116$\times$ more atoms compared to the PSC algorithm.
This advantage over the PSC algorithm is evident in the asymptotic scaling behavior shown in \Cref{subfig:grid_op_count}.
By fully reflecting the underlying hardware capabilities, the proposed algorithm achieves $\bigO(N)$ scaling through the parallel transport of multiple atoms, whereas the PSC algorithm is fundamentally limited to $\bigO(\sqrt N)$ parallelism.
As a result, the slope of the proposed algorithm in \Cref{subfig:atom_moved_parallelism} is approximately twice as large as that of the PSC algorithm.
This massive degree of parallelism is the primary reason for the shorter reconfiguration time observed in \Cref{subfig:grid_reconfiguration_time} because transporting $N$ atoms in parallel provides an $N$-fold increase in transportation throughput.

One potential concern with the proposed algorithm is the increase in the travel distance required for each atom, as shown in \Cref{subfig:atom_distance}, where for sufficiently large $N$ the per-atom travel distance exceeds that of the PSC algorithm by at most a factor of 2.75; 
we consider this overhead to be acceptable because the travel distance per atom follows $\bigO(n)$ scaling.
Therefore, the proposed algorithm effectively exploits massive parallelism while incurring a reasonable increase in travel distance.

Finally, we analyzed the number of operations per atom.
\Cref{subfig:grid_avg_moved_atoms} shows that the proposed method increases the number of operations with $\bigO(n)$ linear scaling, whereas the PSC algorithm maintains an approximately constant number of operations close to two.
This increase highlights a limitation of the proposed algorithm because it prioritizes minimizing total reconfiguration time through massive parallelism at the cost of subjecting individual atoms to more frequent manipulations.
Because a larger number of capture-and-release cycles can increase the risk of atom loss owing to heating, addressing this limitation remains an important direction for future research.

\paragraph{Detailed comparison with PSC.}
We compared our leftward alignment algorithm (\Cref{alg:left-align}) with the corresponding component in the PSC algorithm, namely the Tetrimino construction algorithm.
Both algorithms perform row-by-row compression of atoms, which largely determines overall reconfiguration performance.
Compression performance was evaluated using two representative tasks:
\begin{itemize}
    \item Gathering task. The algorithm gathers atoms toward the center of each row.
    \item Randomizing task. The algorithm moves atoms to randomized positions within each row.
\end{itemize}

As shown in \Cref{subfig:gather_rowwise_reconf_time}, the proposed algorithm achieves a smaller total transportation cost than the PSC algorithm for gathering tasks.
Similar trends are observed for randomizing tasks in \Cref{subfig:random_rowwise_reconf_time}, although the performance gap between the proposed and PSC algorithms becomes smaller while overall scalability is maintained.

The high efficiency of the proposed algorithm primarily stems from its strong compression performance because the total transportation cost observed in \Cref{subfig:grid_movement_cost} is nearly identical to that in \Cref{subfig:gather_rowwise_reconf_time}.
This indicates that the overall reconfiguration process is largely dominated by the leftward alignment step, which in turn determines the reconfiguration time shown in \Cref{subfig:grid_reconfiguration_time}.
Therefore, the performance of the reconfiguration plan is largely determined by the effectiveness of row-wise compression.

\subsection{Evaluation of planning algorithms}

We next evaluated the proposed algorithm in terms of planning time and reliability, where reliability is measured by the success rate, with a success rate of 100\% indicating that the planning algorithm always produces a valid reconfiguration plan for the given target geometries.
As shown in \Cref{subfig:grid_planning_time}, the proposed algorithm and the PSC algorithm exhibit similar planning times for grid geometries, while the Hungarian algorithm shows a prohibitive increase in planning time as $N$ increases.
For arbitrary geometries, \Cref{subfig:arbitrary_planning_time} shows that the planning time of the proposed method is approximately 36\%--54\% longer than that of the PSC algorithm.
Overall, these results confirm that the proposed algorithm consistently runs successfully within a reasonable planning time.

\Cref{subfig:grid_success_rate} shows that both the proposed algorithm and the Hungarian algorithm consistently produce valid reconfiguration plans for grid formation problems.
A key design principle of the proposed algorithm is the use of a fallback mechanism to handle potential failures.
In most cases, the algorithm applies the two-step strategy to achieve efficiency comparable to that of the PSC algorithm, while falling back to the three-step strategy when necessary to guarantee correctness.
This fallback mechanism ensures that the proposed algorithm always generates a valid reconfiguration plan, even for challenging edge cases, and the experimental results confirm its effectiveness in practical grid formation scenarios.
In contrast, the PSC algorithm fails to find a solution in some grid formation instances owing to the heuristic limitations of its two-step strategy.

For arbitrary formation problems, all compared algorithms achieved a $100\%$ success rate, as shown in \Cref{subfig:arbitrary_success_rate}.
These results help explain why the PSC algorithm exhibits lower success rates when $N < 10^3$ in \Cref{subfig:grid_success_rate}.
The reduced success rate can be attributed to the fact that the right-hand side of \Cref{eq:existence_iff}, corresponding to the Gale--Ryser inequality, tends to be smaller for a small number $N$ of atoms, which causes the inequality to fail more frequently under high-variance conditions.
In particular, a higher atom density leads to greater variance in row sums and column sums, which increases the value of the left-hand side of \Cref{eq:existence_iff} and makes the condition more difficult to satisfy, especially when the right-hand side is tightly constrained.
Overall, the observations indicate that the proposed algorithm remains robust and effective for arbitrary target geometries and, moreover, demonstrates particularly strong performance for grid target geometries.

Finally, we evaluated the effects of the three-step strategy and the peephole optimization.
\Cref{subfig:grid_reconf_three_step_ratio} shows that up to 4\% of grid formation problems could not be solved without using the three-step strategy.
However, the three-step strategy was not applied to arbitrary formation problems, as shown in \Cref{subfig:arbitrary_reconf_three_step_ratio}.
This demonstrates that the three-step strategy is essential to guarantee a valid reconfiguration plan for any problem.
Regarding the impact of peephole optimization, the blue and orange lines in \Cref{subfig:grid_peephole_optimization} show that it effectively reduces the number of operations by approximately 32$\%$.
Without the peephole optimization, the grid formation strategy exhibits performance similar to that of the two-step strategy used for arbitrary formation.
Therefore, the peephole optimization is a key component for achieving high efficiency in grid formation problems.

\section{Conclusion}
In this study, we proposed a planning algorithm that efficiently creates large-scale defect-free atom arrays.
For a system with $N$ atoms, the algorithm produces a reconfiguration plan that achieves the target atom geometry in $\bigO(\sqrt{N})$ time by exploiting a 2D lattice pattern generated by a two-axis AOD.

The algorithm is built on a divide-and-conquer framework that leverages inherent hardware parallelism, in which the decomposer first splits the problem into at most three 1D shuttling tasks, and the 1D shuttling solver then executes each task in $\bigO(\sqrt{N})$ time by transporting $\bigO(N)$ atoms in parallel.
The high efficiency of the approach primarily arises from the 1D shuttling solver, which constructs a left-aligned atom geometry using simple row-wise operations.
By relying on the Gale--Ryser theorem~\cite{Gale1957-kk,Ryser1957-kk}, the proposed algorithm guarantees a highly reliable solution for arbitrary target geometries.
In addition, we introduce a peephole optimization technique that further improves reconfiguration efficiency for grid target geometries.

Numerical simulations showed substantial reductions in total transportation cost compared to previous SOTA algorithms, highlighting the proposed method's potential for large-scale quantum systems.
The total transportation cost decreases from $\bigO(N)$ to $\bigO(\sqrt{N})$, while the operation count remains at $\bigO(\sqrt{N})$, consistent with theoretical analysis.
For instance, with $N=2 \times 10^5$, the total transportation cost of the proposed method is only 1/7 of that of the prior SOTA algorithms.
The number of operations in the proposed method maintains $\bigO(\sqrt{N})$ scaling, matching that of the previous algorithms.

One direction for future work is to address atom loss during reconfiguration, a critical factor for realizing truly scalable systems.
To overcome this limitation and to enable long-duration operation of large-scale systems, several groups have recently developed schemes for continuous and repeated loading of atoms from auxiliary reservoirs~\cite{cite:lukin_3000, arXiv_2506_15633, cite:iterative_assembly, cite:optical_lattices}, where efficient atom reconfiguration between a storage and a reservoir is required.
Extending the present approach to such architectures is an interesting direction.

\section*{Acknowledgments}
This work was supported by JSPS KAKENHI Grant Numbers 23K18464 and 23K28061, MEXT Quantum Leap Flagship Program (MEXT Q-LEAP) Grant Number JPMXS0118069021, and JST Moonshot R\&D Program Grant Number JPMJMS2269.

\bibliographystyle{unsrtnat}
\bibliography{citation}

\onecolumn\newpage
\appendix

\section{General Form of 2D Tweezer System} \label{appendix:difficult_operations}
We present a general form of the 2D tweezer system used in previous studies~\cite{cite:atom_movr,cite:multiple_tweezer_kagome,cite:multiple_tweezer_grid}.
This general form, $(n,S,\Phi',F')$, was used to evaluate the previous algorithms in \Cref{sec:evaluation}.
This previous model assumes that the tweezer system supports long-distance transportation operations, in which atoms can move from their initial sites to target sites through multiple relay points, and additionally allows captured atoms to move in different directions simultaneously.
We introduce these flexible but more complex operations as an alternative to our model $(n,S,\Phi,F)$, which uniformly transports captured atoms in the same direction.

\paragraph{2D Tweezer Operations $\Phi'$.}
A primitive operator $\phi \in \Phi': S \rightarrow S$ is represented as a 3-tuple $(I_0, J_0, \Delta)$, where $I_0 \subseteq X$ and $J_0 \subseteq X$ denote the initial sets of rows and columns, respectively, that define a lattice pattern generated by the mobile tweezers.
The vector $\Delta = ( \delta_1, \delta_2, \dots, \delta_{|\Delta|} )$ specifies the sequence of $|\Delta| \in \mathbb{N}$ moves required by the operator $\phi$.
The $m$-th move $\delta_m: X^2 \rightarrow X^2$, for $1 \leq m \leq |\Delta|$, is defined as follows:
\begin{equation}\label{eq:axismove}
\delta_m(i,j) \mapsto (\delta_m^\mathrm{(row)}(i),\ \delta_m^\mathrm{(col)}(j)),
\end{equation}
where $\delta_m^\mathrm{(row)}: X \rightarrow X$ and $\delta_m^\mathrm{(col)}: X \rightarrow X$ are the mappings that return the set of rows and columns after moving, respectively.
An operator $\phi$ transports the atom from the initial site $(i_0,j_0) \in I_0 \times J_0$ to the target site $(i_{|\Delta|}, j_{|\Delta|})$ via $(|\Delta|-1)$ relay points according to the following recurrence:
\begin{equation}\label{eq:atomPos}
(i_m,j_m) = \delta_m(i_{m-1},j_{m-1}),
\end{equation}
where $1 \leq m \leq |\Delta|$.
Accordingly, the mobile tweezers move the lattice pattern according to the following recurrence:
\begin{equation} \label{eq:tweezerPos}
    (I_m, J_m) =
        ( \{\delta^{(\mathrm{row})}_m(i) \mid i \in I_{m-1}\},       \{\delta^{(\mathrm{col})}_m(j) \mid j \in J_{m-1}\} ),
\end{equation}
where $1 \leq m \leq |\Delta|$.

Collision-free atom transportation can be ensured by imposing the following two constraints, which prevent (1) collisions between moving atoms and (2) collisions between moving and stationary atoms.
\begin{description}
    \item[Constraint 1: Intra-tweezer order preservation.]
The mobile tweezer must maintain strictly increasing row and column indices to prevent overtaking between moving atoms.
Formally, the tweezer operation must satisfy:
\begin{align}
    \forall m \in \{1, 2, \dots, |\Delta|\}, \forall i',i'' \in I_{m-1}: i' < i'' \implies \delta_m^\mathrm{(row)}(i') < \delta_m^\mathrm{(row)}(i''), \\
    \forall m \in \{1, 2, \dots, |\Delta|\}, \forall j', j'' \in J_{m-1}: j' < j'' \implies \delta_m^\mathrm{(col)}(j') < \delta_m^\mathrm{(col)}(j''). \label{eq:colOrder}
\end{align}

    \item[Constraint 2: Collision-free routing.]
The mobile tweezer must avoid collisions between captured and uncaptured atoms.
Specifically, an atom initially at site $(i_0, j_0) \in (I_0 \times J_0)$ must be transported along a collision-free path, where all intermediate sites are vacant:
\begin{equation}\label{eq:duplicate}
\forall (i,j) \in \bigcup_{m=0}^{|\Delta|} (I_m \times J_m) \setminus (I_0 \times J_0): A_{i,j} = 1 \implies A_{i_0,j_0} = 0.
\end{equation}
\end{description}

For a tweezer operation $\phi$ that satisfies the abovementioned constraints, the transportation cost $T(\phi)$ is defined as follows:
\begin{align}
    T(\phi) &= \sum_{m = 1}^{|\Delta|} \sqrt{\mathrm{dist}(\delta_m)}, \label{eq:sqrt_complex}
\end{align}
where $\textrm{dist}(\delta_m)$ is the maximum distance traveled by any atom during the $m$-th move $\delta_m$, defined as
\begin{equation}
  \textrm{dist}(\delta_m) = \mathrm{max}_{(i,j) \in I_{m-1} \times J_{m-1}}\left(\sqrt{\left|i - \delta_m^\mathrm{(row)}(i)\right|^2 + \left|j - \delta_m^\mathrm{(col)}(j)\right|^2} \right).
\end{equation}
The sqrt cost model in \Cref{eq:sqrt_complex} assumes that transportation speed is primarily limited by acceleration~\cite{cite:sqrt_time_hardware1, cite:multiple_tweezer_kagome, cite:single_tweezer_3}.
In other words, this model abstracts shuttling protocols in which atoms are transported under constant acceleration.

An alternative to the sqrt cost model is the linear cost model:
\begin{align}
    T(\phi) &= \sum_{m = 1}^{|\Delta|} \mathrm{dist}(\delta_m). \label{eq:linear_complex}
\end{align}
The linear cost model assumes that transportation speed is primarily limited by velocity~\cite{cite:Hungarian, cite:logical_quantum_processor, graham2022multi}.
This model represents shuttling protocols in which atoms are moved at a constant velocity.

Note that, for the proposed algorithm, the linear cost model is equivalent to the sqrt cost model because their objective functions coincide when the transportation distance is $1$, i.e., $\mathrm{dist}(\cdot)=1$.
Owing to this equivalence, the objective function can be expressed in a simplified form, as given in \Cref{eq:object}.

\paragraph{Objective Function $F'$.}

We estimate the reconfiguration time of a plan $p = (\phi_1, \phi_2, \dots, \phi_{|p|})$ using the following objective function:
\begin{equation}\label{eq:object2}
    F'(p, t_1, t_2) = |p| \, t_1 + \sum_{k=1}^{|p|} T(\phi_k) \, t_2,
\end{equation}
where the first and second terms represent the operation cycle time and the transportation time, respectively (see \Cref{subsec:model}). In the previous model, the total transportation cost is given by $\sum_{k=1}^{|p|} T(\phi_k)$, as shown in \Cref{eq:object2}.

In contrast to the complex operations described above, each 2D tweezer operation $\phi$ in our model incurs a fixed transportation cost of $T(\phi) = 1$ owing to $\mathrm{dist}(\cdot)=1$.
Therefore, the total transportation cost in our model equals the number $|p|$ of 2D tweezer operations, as detailed in \Cref{subsec:model}.

\section{Proof of Correctness of the 1D Shuttling Solver}
\label{appendix:1d_shuttling_solver}

We present a proof of the correctness of our 1D shuttling solver.
First, we adopt the convention $\forall i \in X: A_{i, n+1} = 0$ and define a parameterized version of the left-aligned property.
\begin{definition}
Let $A \in S$ be a geometry and $x$ a non-negative integer.
We say that $A$ is \textit{$x$-partially left-aligned} if it satisfies the inequality in \Cref{eq:partially-left-aligned-property}.
\begin{equation}\label{eq:partially-left-aligned-property}
    \forall i \in X, \forall j \in \{x,x + 1,\dots,n\}: A_{i,j} \geq A_{i, j+1}.
\end{equation}
\end{definition}
Note that the left-aligned property is equivalent to the $1$-partially left-aligned property.

We now show that any plan generated by \Cref{alg:left-align} transforms the initial atom geometry $A^\mathrm{(int)}$ into the left-aligned geometry $A^\mathrm{(left)}$.
\begin{lemma}\label{lemma:left-align}
Let $A^{(\mathrm{int})}$ be an arbitrary geometry, and $p$ be a plan output by \Cref{alg:left-align}.
Then the resulting geometry $A^\mathrm{(left)} = p(A^{(\mathrm{int})})$ is left-aligned.
\end{lemma}

\begin{proof}
We proceed by mathematical induction.
Let $Q_1(x)$ denote the statement that the intermediate geometry $A^{(\gamma)}$ is $x$-partially left-aligned at the end of loop $x$ in \Cref{alg:left-align} (lines 2--6).

\textbf{Base step:}
For $x = n$, the statement $Q_1(n)$ holds because $A_{i, n+1} = 0$.

\textbf{Induction step:}
Assume that $Q_1(x)$ holds for some integer $x \ge 2$.
Let $A^{(\gamma)}$ be the geometry at the end of loop $x$ (and thus the start of loop ($x-1$)), and let $A^{(\gamma^\prime)}$ be the geometry at the end of loop $x-1$.
We aim to show that $Q_1(x-1)$ also holds, i.e., that
\begin{equation}\label{appendix:left-align-next}
  \forall (i,j) \in X^2: j \geq x-1 \implies A^{(\gamma^\prime)}_{i,j} \geq A^{(\gamma^\prime)}_{i,j+1}. 
\end{equation}

We consider the effect of the $\mathrm{Left}(I,J)$ operation applied during loop $x-1$ (line 5) for an arbitrary row~$i$.
There are two possibilities for the geometry $A^{(\gamma)}$ before this operation:

\textbf{(i) Case $A^{(\gamma)}_{i,x-1} = 0$: }
If the element at column $x-1$ is 0, the $\mathrm{Left}(I,J)$ operation shifts the portion of row $i$ from column $x$ onwards one position to the left. The row $i$ of new geometry $A^{(\gamma^\prime)}$ is then defined as
\begin{equation}
A^{(\gamma^\prime)}_{i,j} = 
\begin{cases}
A^{(\gamma)}_{i,j}, & \text{if $1 \leq j < x-1$,} \\
A^{(\gamma)}_{i,j+1}, & \text{if $x-1 \leq j < n$,} \\
0, & \text{if $j = n$.}
\end{cases}
\end{equation}
We must verify that $A^{(\gamma^\prime)}_{i,j} \geq A^{(\gamma^\prime)}_{i,j+1}$ holds for all $j \geq x-1$.
\begin{itemize}
    \item For $x-1 \leq j < n-1$: We have $A^{(\gamma^\prime)}_{i,j} = A^{(\gamma)}_{i,j+1}$ and $A^{(\gamma^\prime)}_{i,j+1} = A^{(\gamma)}_{i,j+2}$. Because $j+1 \geq x$, the inductive hypothesis $Q_1(x)$ guarantees that $A^{(\gamma)}_{i,j+1} \geq A^{(\gamma)}_{i,j+2}$. Therefore, $A^{(\gamma^\prime)}_{i,j} \geq A^{(\gamma^\prime)}_{i,j+1}$.
    \item For $j = n-1$: We have $A^{(\gamma^\prime)}_{i,n-1} = A^{(\gamma)}_{i,n}$ and $A^{(\gamma^\prime)}_{i,n} = 0$. Because $A^{(\gamma)}_{i,n}$ is either 0 or 1, the inequality $A^{(\gamma^\prime)}_{i,n-1} \geq A^{(\gamma^\prime)}_{i,n}$ holds.
\end{itemize}
Thus, case (i) satisfies the condition in \Cref{appendix:left-align-next}.

\textbf{(ii) Case $A^{(\gamma)}_{i,x-1} = 1$:}
If the element at column $x-1$ is $1$, the operation leaves row $i$ unchanged.
Thus, we have $A^{(\gamma^\prime)}_{i,j} = A^{(\gamma)}_{i,j}$ for $j \geq x-1$. 
We then need to show that $A^{(\gamma)}_{i,j} \geq A^{(\gamma)}_{i,j+1}$ for $j \geq x-1$.
\begin{itemize}
    \item For $j \geq x$: This follows directly from the inductive hypothesis $Q_1(x)$.
    \item For $j = x-1$: We need $A^{(\gamma)}_{i,x-1} \geq A^{(\gamma)}_{i,x}$. Because $A^{(\gamma)}_{i,x-1} = 1$ in this case, and $A^{(\gamma)}_{i,x}$ is either $0$ or $1$, the inequality $1 \geq A^{(\gamma)}_{i,x}$ always holds.
\end{itemize}
Thus, case (ii) satisfies the condition in \Cref{appendix:left-align-next}.

\textbf{Conclusion:}
We have shown that if $Q_1(x)$ holds, then $Q_1(x-1)$ also holds.
Because the base step $Q_1(n)$ is satisfied, the principle of mathematical induction implies that $Q_1(1)$ holds.
As the algorithm terminates after the last loop ($x=1$), the final output $A^\mathrm{(left)}$ satisfies $Q_1(1)$.
Because the left-aligned property is equivalent to the $1$-partially left-aligned property, this establishes that $A^\mathrm{(left)}$ is indeed left-aligned, as required.
\end{proof}

Using an analogous inductive argument, but proceeding in the opposite direction, it can also be shown that the output plan generated by \Cref{alg:drop-off} transforms the left-aligned geometry $A^\mathrm{(left)}$ into the target geometry $A^\mathrm{(tgt)}$.
\begin{lemma}\label{lemma:inverse_left_align}
Let $A^{(\mathrm{tgt})}$ be an arbitrary geometry, and let $p$ be a plan output by \Cref{alg:drop-off} with input $A^{(\mathrm{tgt})}$.
Let $A^\mathrm{(left)}$ be a left-aligned geometry that satisfies
\begin{equation}\label{eq:row_sums_same}
    \forall i\in X:\sum^{n}_{j=1} A^{(\mathrm{left})}_{i,j} = \sum^{n}_{j=1} A^{(\mathrm{tgt})}_{i,j}.
\end{equation}
Then, \Cref{eq:lemma_dropoff_result} holds.
\begin{equation}\label{eq:lemma_dropoff_result}
    p(A^\mathrm{(left)}) = A^\mathrm{(tgt)}
\end{equation}
\end{lemma}

\begin{proof}
We proceed by mathematical induction.
Let $A^{(\gamma)}$ denote the geometry at the start of loop $x$, and $A^{(\gamma^\prime)}$ the geometry at the end of loop $x$ (and thus the start of loop $x+1$).
Let $Q_2(x)$ be the statement that both of the following conditions hold at the start of loop $x$ in \Cref{alg:drop-off} (lines 2--6).
\begin{itemize}
    \item The intermediate geometry $A^{(\gamma)}$ is $x$-partially left-aligned.
    \item The left $x-1$ columns of $A^{(\gamma)}$ match those of $A^\mathrm{(tgt)}$ as follows.
    \begin{equation}\label{eq:left-matches}
        \forall i \in X,\forall j \in \{1,2,\dots,x-1\}: A^{(\gamma)}_{i,j} = A^\mathrm{(tgt)}_{i,j}
    \end{equation}
\end{itemize}

\textbf{Base step:}
For $x = 1$, the statement $Q_2(1)$ holds by definition.

\textbf{Induction step:}
Assume that $Q_2(x)$ holds for some integer $x \ge 1$. 
We aim to show that $Q_2(x+1)$ also holds, which consists of two claims.
\begin{enumerate}
    \item $A^\mathrm{(\gamma^\prime)}$ is $(x+1)$-partially left-aligned.
    \item The left $x$ columns of $A^\mathrm{(\gamma^\prime)}$ match those of $A^\mathrm{(tgt)}$ as follows: \begin{equation}\label{eq:partial_match_x+1}
        \forall i \in X,\forall j \in \{1,2,\dots,x\}: A^{(\gamma^\prime)}_{i,j} = A^\mathrm{(tgt)}_{i,j}
    \end{equation}
\end{enumerate}

We consider the effect of the $\mathrm{Right}(I,J)$ operation applied during loop $x$ (line 5) for an arbitrary row~$i$.
There are two possibilities for the geometry $A^{(\mathrm{tgt})}$ before this operation:

\textbf{(i) Case $A^{(\mathrm{tgt})}_{i,x} = 0$: }
The $\mathrm{Right}(I,J)$ operation shifts the portion of row $i$ from column $x$ onwards one position to the right. The row $i$ of new geometry $A^{(\gamma^\prime)}$ is then defined as
\begin{equation}
A^{(\gamma^\prime)}_{i,j} = 
\begin{cases}
A^{(\gamma)}_{i,j}, & \text{if $1 \leq j < x$,} \\
0, & \text{if $j = x$,} \\
A^{(\gamma)}_{i,j-1}, & \text{if $x < j \leq n$.}
\end{cases}
\end{equation}
With the precondition $Q_2(x)$, we can confirm that $A^{(\gamma^\prime)}$ is $(x+1)$-partially left-aligned and satisfies \Cref{eq:partial_match_x+1}.
Thus, case (i) satisfies $Q_2(x+1)$.

\textbf{(ii) Case $A^{(\mathrm{tgt})}_{i,x} = 1$:}
According to \Cref{alg:drop-off}, the operation leaves row $i$ unchanged.
Thus, $A^{(\gamma^\prime)}_{i,j} = A^{(\gamma)}_{i,j}$.
Hence, $A^{(\gamma^\prime)}$ is $(x+1)$-partially left-aligned.
We now need to verify that \Cref{eq:partial_match_x+1} holds for $A^{(\gamma^\prime)}$.
We prove this by contradiction.
Suppose that $A^{(\gamma)}_{i, x} = 0$.
Because $A^{(\gamma)}$ is $x$-partially left-aligned, \Cref{eq:tempo_1} holds.
\begin{equation}\label{eq:tempo_1}
    \forall j \in \{x,x+1, \dots,n\}: A^{(\gamma)}_{i,j} = 0.
\end{equation}
Because the left $x-1$ columns of $A^{(\gamma)}$ match those of $A^\mathrm{(tgt)}$, \Cref{eq:tempo_2} holds.
\begin{equation}\label{eq:tempo_2}
    \sum^n_{j=1} A^{(\gamma)}_{i,j} < \sum^n_{j=1} A^{(\mathrm{tgt})}_{i,j}
\end{equation}
This contradicts \Cref{eq:row_sums_same}.
Therefore, $A^{(\gamma)}_{i, x} = 1$ holds.
This means that \Cref{eq:partial_match_x+1} holds for $A^{(\gamma^\prime)}$ because row $i$ of $A^{(\gamma^\prime)}$ remains unchanged from $A^{(\gamma)}$.
Thus, case (ii) satisfies $Q_2(x+1)$.

\textbf{Conclusion:}
We have shown that if $Q_2(x)$ holds, then $Q_2(x+1)$ also holds.
Because the base step $Q_2(1)$ is satisfied, the principle of mathematical induction implies that $Q_2(n+1)$ holds.
Because the algorithm terminates after the last loop ($x=n$), the final output $p(A^\mathrm{(left)})$ satisfies $Q_2(n+1)$, which implies that $p(A^{\mathrm{(left)}})$ is equal to $A^{(\mathrm{tgt})}$.
This completes the proof.
\end{proof}

We are now ready to present the main theorem.

\begin{theorem}[1D Shuttling solver]\label{thm:1d_shuttling_problem}
    Let $A^\mathrm{(int)}$ and $A^\mathrm{(tgt)}$ be geometries that satisfy at least one of \Cref{eq:row_sum_eq_condition} or \Cref{eq:column_sum_eq_condition}.
    \begin{align}
        \forall i \in X: \sum^n_{j=1} A^{\mathrm{(int)}}_{i,j} = \sum^n_{j=1} A^{\mathrm{(tgt)}}_{i,j}, \label{eq:row_sum_eq_condition} \\
        \forall j \in X: \sum^n_{i=1} A^{\mathrm{(int)}}_{i,j} = \sum^n_{i=1} A^{\mathrm{(tgt)}}_{i,j} .\label{eq:column_sum_eq_condition}        
    \end{align}
    Then, there exists a feasible plan $p \in P$ that converts $A^{\mathrm{(int)}}$ to $A^\mathrm{(tgt)}$ with $|p| = 2(n-1)$.
\end{theorem}
\begin{proof}
Without loss of generality, we consider the case where \Cref{eq:row_sum_eq_condition} is satisfied.
The proof for the other case, \Cref{eq:column_sum_eq_condition}, follows similarly by rotating the matrix by 90\textdegree.

Let $p_1$ and $p_2$ be the plans generated by \Cref{alg:left-align} with input $A^\mathrm{(int)}$ and by \Cref{alg:drop-off} with input $A^\mathrm{(tgt)}$, respectively.
We construct $A^\mathrm{(left)}$ such that $A^\mathrm{(left)} = p_1(A^\mathrm{(int)})$.
From \Cref{lemma:left-align}, the resulting geometry $A^\mathrm{(left)}$ is left-aligned.
Moreover, from \Cref{eq:row_sum_eq_condition}, $A^\mathrm{(left)}$ satisfies \Cref{eq:meet_in_the_middle}.
\begin{equation}\label{eq:meet_in_the_middle}
    \forall i\in X:\sum^{n}_{j=1} A^{(\mathrm{left})}_{i,j} = \sum^{n}_{j=1} A^{(\mathrm{tgt})}_{i,j}.
\end{equation}
Hence, from \Cref{lemma:inverse_left_align}, the composed plan $p_2 \circ p_1$ satisfies \Cref{eq:finalize}.
\begin{equation}\label{eq:finalize}
    p_2 \circ p_1 (A^\mathrm{(int)}) = p_2 (A^\mathrm{(left)}) = A^\mathrm{(tgt)}.
\end{equation}
Therefore, $p_2 \circ p_1$ is a feasible plan that converts $A^{\mathrm{(int)}}$ to $A^\mathrm{(tgt)}$.
From \Cref{alg:left-align,alg:drop-off}, we have $|p_2 \circ p_1| = 2(n-1)$.
This completes the proof.
\end{proof}

\section{Proof of Correctness of the Decomposer}
\label{appendix:three-step-strategy}

We show that the three-step strategy always decomposes the problem into three 1D shuttling tasks.
From \Cref{alg:leveling}, it follows directly that \Cref{corollary:leveling_prerequirements} holds.
\begin{corollary}\label{corollary:leveling_prerequirements}
    For any geometry $A^\mathrm{(int)} \in S$, there exists a geometry $A^\mathrm{(rbal)}$ satisfying 
    \begin{align}
        \forall j \in X&: \sum^{n}_{i=1} A^\mathrm{(int)}_{i,j} = \sum^n_{i=1} A^\mathrm{(rbal)}_{i,j}, \\
        \forall i^\prime, i^{\prime\prime} \in X&: \left| \sum^{n}_{j=1} A^\mathrm{(rbal)}_{i^\prime,j} - \sum^{n}_{j=1} A^\mathrm{(rbal)}_{i^{\prime\prime},j} \right| \leq 1.
    \end{align}
\end{corollary}

The correctness of the three-step strategy relies primarily on \Cref{lemma:existence_cfin}.
\begin{lemma}\label{lemma:existence_cfin}
Let $A^\mathrm{(tgt)}$ be an arbitrary geometry, and let $A^\mathrm{(rbal)}$ be a geometry satisfying
\begin{align}
    & \sum^n_{i=1} \sum^n_{j=1} A^\mathrm{(rbal)}_{i,j} = \sum^n_{i=1} \sum^n_{j=1} A^\mathrm{(tgt)}_{i,j}, \label{eq:D_3} \\
    \forall i^\prime, i^{\prime\prime} \in X:& \left| \sum^{n}_{j=1} A^\mathrm{(rbal)}_{i^\prime,j} - \sum^{n}_{j=1} A^\mathrm{(rbal)}_{i^{\prime\prime},j} \right| \leq 1. \label{eq:D_4}
\end{align}

Then, there exists a geometry $A^\mathrm{(cfin)}$ that satisfies
\begin{align}
    \forall i \in X&: \sum^n_{j=1} A^\mathrm{(cfin)}_{i,j} = \sum^n_{j=1} A^\mathrm{(rbal)}_{i,j}, \label{eq:D_5} \\
    \forall j \in X&: \sum^n_{i=1} A^\mathrm{(cfin)}_{i,j} = \sum^n_{i=1} A^\mathrm{(tgt)}_{i,j}. \label{eq:D_6}
\end{align}
\end{lemma}

\begin{proof}
    It suffices to verify that the Gale--Ryser inequality (\Cref{eq:existence_iff}) holds under the conditions of $A^\mathrm{(cfin)}$.
    From \Cref{eq:D_4}, for any non-negative number $x$, we have
\begin{equation}
     \sum^{n}_{i=1} \mathrm{min}\left( \sum^{n}_{j=1} A^\mathrm{(rbal)}_{i,j}, x \right) = \mathrm{min}\left(xn, \sum^n_{i=1} \sum^n_{j=1} A^\mathrm{(rbal)}_{i,j}\right).
\end{equation}
    From the definition of geometry, for any non-negative number $x$, we have
\begin{equation}
     \sum^{x}_{j=1} \sum^{n}_{i=1} A^\mathrm{(tgt)}_{i,j} \leq \mathrm{min}\left(xn, \sum^n_{i=1} \sum^n_{j=1} A^\mathrm{(tgt)}_{i,j}\right).
\end{equation}
    Using \Cref{eq:D_3}, we obtain
\begin{equation}
    \sum^{x}_{j=1} \sum^{n}_{i=1} A^\mathrm{(tgt)}_{i,j} \leq \sum^{n}_{i=1} \mathrm{min}\left( \sum^{n}_{j=1} A^\mathrm{(rbal)}_{i,j}, x \right).
\end{equation}
    Therefore, from \Cref{theorem:gale_ryser}, there exists a geometry $A^\mathrm{(cfin)}$ that satisfies \Cref{eq:D_5,eq:D_6}.
    This geometry $A^\mathrm{(cfin)}$ can be constructed using \Cref{alg:intermediate}.
\end{proof}

We now demonstrate the correctness of the three-step strategy.
\begin{theorem}
    Let $A^\mathrm{(int)}$ and $A^\mathrm{(tgt)}$ be geometries satisfying
    \begin{equation}
        \sum^n_{i=1} \sum^n_{j=1} A^\mathrm{(int)}_{i,j} = \sum^n_{i=1} \sum^n_{j=1} A^\mathrm{(tgt)}_{i,j}.
    \end{equation}
    Then, there exist geometries $A^\mathrm{(rbal)}$ and $A^\mathrm{(cfin)}$ that satisfy
    \begin{align}
        \forall j \in X&: \sum^{n}_{i=1} A^\mathrm{(int)}_{i,j} = \sum^n_{i=1} A^\mathrm{(rbal)}_{i,j}, \label{eq:D_12} \\
        \forall i \in X&: \sum^{n}_{j=1} A^\mathrm{(rbal)}_{i,j} = \sum^n_{j=1} A^\mathrm{(cfin)}_{i,j}, \label{eq:D_13}  \\
        \forall j \in X&: \sum^{n}_{i=1} A^\mathrm{(cfin)}_{i,j} = \sum^n_{i=1} A^\mathrm{(tgt)}_{i,j}.  \label{eq:D_14}
    \end{align}
\end{theorem}

\begin{proof}
    We first obtain $A^\mathrm{(rbal)}$ by executing \Cref{alg:leveling} with input $A^\mathrm{(int)}$.
    From \Cref{corollary:leveling_prerequirements}, $A^\mathrm{(rbal)}$ satisfies \Cref{eq:D_12}.
    Let $R$ be the row sums of $A^\mathrm{(rbal)}$ and $C$ be the column sums of $A^\mathrm{(tgt)}$.
    From \Cref{lemma:existence_cfin}, we can construct $A^\mathrm{(cfin)}$ satisfying \Cref{eq:D_13,eq:D_14} by executing \Cref{alg:intermediate} with inputs $R$ and $C$.
    This completes the proof.
\end{proof}

\end{document}